\documentclass[a4paper,USenglish,cleveref,pdfa,authorcolumns,preprint]{lipics-v2021}

\pdfoutput=1
\hideLIPIcs
\nolinenumbers

\usepackage{preamble}

\title{Structural Parameters for Steiner Orientation}

\titlerunning{Structural Parameters for Steiner Orientation}

\author{Tesshu Hanaka}
{Kyushu University, Fukuoka, Japan
\and \url{https://sites.google.com/view/tesshu-hanaka/home}}
{hanaka@inf.kyushu-u.ac.jp}
{https://orcid.org/0000-0001-6943-856X}
{JSPS KAKENHI Grant Numbers
JP21K17707, 
JP22H00513, 
JP23H04388, 
JP25K03077, 
and JST, CRONOS, Japan Grant Number JPMJCS24K2.
}

\author{Michael Lampis}{Universit\'{e} Paris-Dauphine, PSL University, CNRS UMR7243, LAMSADE, Paris, France}{michail.lampis@dauphine.fr}{https://orcid.org/0000-0002-5791-0887}
{Partially supported by ANR project ANR-21-CE48-0022 (S-EX-AP-PE-AL).}

\author{Nikolaos Melissinos}{Computer Science Institute, Faculty of Mathematics and Physics, Charles University, Prague, Czech Republic}{melissinos@iuuk.mff.cuni.cz}{https://orcid.org/0000-0002-0864-9803}{Partially supported by Charles University projects UNCE 24/SCI/008 and PRIMUS 24/SCI/012, and by the project 25-17221S of GAČR.}

\author{Edouard Nemery}
{Universit\'{e} Paris-Dauphine, PSL University, CNRS UMR7243, LAMSADE, Paris, France}
{edouard.nemery@etu.u-paris.fr}
{https://orcid.org/0009-0007-6977-9330}{}

\author{Hirotaka Ono}
{Nagoya University, Nagoya, Japan}
{}
{JSPS KAKENHI Grant Numbers
JP22H00513, 
}{}{}

\author{Manolis Vasilakis}
{Universit\'{e} Paris-Dauphine, PSL University, CNRS UMR7243, LAMSADE, Paris, France}
{emmanouil.vasilakis@dauphine.eu}
{https://orcid.org/0000-0001-6505-2977}
{Partially supported by ANR project ANR-21-CE48-0022 (S-EX-AP-PE-AL).}

\authorrunning{J. Open Access and J.\,R. Public} 

\Copyright{Jane Open Access and Joan R. Public} 



\category{} 

\EventEditors{John Q. Open and Joan R. Access}
\EventNoEds{2}
\EventLongTitle{42nd Conference on Very Important Topics (CVIT 2016)}
\EventShortTitle{CVIT 2016}
\EventAcronym{CVIT}
\EventYear{2016}
\EventDate{December 24--27, 2016}
\EventLocation{Little Whinging, United Kingdom}
\EventLogo{}
\SeriesVolume{42}
\ArticleNo{23}

\begin{document}

\maketitle

\begin{abstract}
We consider the \textsc{Steiner Orientation} problem, where we are given as input a mixed graph $G=(V,E,A)$ and a set of $k$ demand pairs $(s_i,t_i)$, $i\in[k]$. The goal is to orient the undirected edges of $G$ in a way that the resulting directed graph has a directed path from $s_i$ to $t_i$ for all $i\in[k]$. We adopt the point of view of structural parameterized complexity and investigate the complexity of \textsc{Steiner Orientation} for standard measures, such as treewidth.
Our results indicate that \textsc{Steiner Orientation} is a surprisingly hard problem from this point of view. In particular, our main contributions are the following:

\begin{enumerate}

\item We show that \textsc{Steiner Orientation} is NP-complete on instances where the underlying graph has feedback
vertex number 2, treewidth 2, pathwidth 3, and vertex integrity 6.

\item We present an XP algorithm parameterized by vertex cover number $\vc$ of complexity $n^{\bO(\vc^2)}$. Furthermore, we show that this running time is essentially optimal by proving that a running time of $n^{o(\vc^2)}$ would refute the ETH.

\item We consider parameterizations by the number of undirected or directed edges ($|E|$ or $|A|$) and we observe that the trivial $2^{|E|}n^{\bO(1)}$-time algorithm for the former parameter is optimal under the SETH. Complementing this, we show that the problem admits a $2^{\bO(|A|)}n^{\bO(1)}$-time algorithm.
    
\end{enumerate}

In addition to the above, we consider the complexity of \textsc{Steiner Orientation} parameterized by $\tw+k$ (FPT), distance to clique (FPT), and $\vc+k$ (FPT with a polynomial kernel).

\end{abstract}

\newpage

\section{Introduction}
The \textsc{Steiner Orientation} problem is an NP-hard graph optimization problem, which involves assigning directions to the undirected edges of a mixed graph—a graph containing both directed and undirected edges—to satisfy specific connectivity requirements, given by pairs of terminals. The objective is to orient the undirected edges in such a way that there exists a directed path between each given terminal pair. Formally, the problem is defined as follows. 

\problemdef{Steiner Orientation} {A mixed graph $G=(V,E,A)$ where $E$ and $A$ denote the set of edges and arcs in $G$ respectively. Additionally, we are given a set of terminal pairs $\terminals = \setdef{(s_i,t_i) \in V \times V}{i \in [k]}$ such that $s_i \neq t_i$, for all $i \in [k]$.} {Determine whether there exists an orientation of $E$ such that $k$ directed paths $P_i = s_i \to \ldots \to t_i$, $i \in [k]$ exist.} 

The \SO{} problem has been well-studied in bioinformatics, motivated by modeling protein–protein or protein–DNA interactions~\cite{MedvedovskyBZS08,jcb/SilverbushES11,DornHKNU11,netmahib/Roayaei19}. Another motivation naturally arises from designing transportation networks. For example, in an urban transportation network with a mix of one-way streets (directed arcs) and two-way streets (undirected edges), one can consider the traffic control problem of deciding the direction of some streets to ensure routes from specific origins to destinations. This scenario is naturally modeled as the \SO{} problem. 

Regarding its complexity, Arkin and Hassin show that \SO{} is NP-complete in general, but polynomially solvable if $k=2$~\cite{ArkinH02}. From the perspective of parameterized complexity, Cygan et al.~\cite{CyganKN13} proposed an $n^{\bO(k)}$-time algorithm for the number of terminal pairs $k$ as a parameter, which shows that the problem belongs to the class XP. However, subsequent research has revealed the problem's intractability for this parameterization; \SO{} is shown to be W[1]-hard when parameterized by $k$ and cannot be solved in time $f(k)\cdot n^{o(k/ \log k)}$ under ETH \cite{PilipczukW18a}. This lower bound was later improved to a $f(k)\cdot n^{o(k)}$ ETH-based lower bound even on planar graphs \cite{ChitnisFS19}. Surprisingly, W{\l}odarczyk recently proved that the problem is not just W[1]-hard, but W[1]-complete~\cite{icalp/Wlodarczyk20}. This series of studies has almost completely characterized the complexity with respect to the parameter $k$. Despite this literature, to the best of our knowledge, the parameterized complexity of \SO{} is only studied for the number $k$ of terminal pairs. The complexity with respect to graph parameters such as treewidth or vertex cover number has been almost entirely unstudied. This paper aims to provide the first systematic study in this direction.

\subsection{Our Contribution} 

In this paper, we first show the para-NP-completeness of \textsc{Steiner Orientation} for the parameters pathwidth, feedback vertex number, and vertex integrity.\footnote{In this paper, graph parameters refer to those of the underlying graph of $G$, that is, the graph obtained by replacing all arcs with undirected edges. We assume that the reader is familiar with standard parameters, but recall some of the relevant definitions and relations further below.} Specifically, the problem is NP-complete on series-parallel graphs (i.e., graphs of treewidth at most 2) of vertex-deletion distance  2 to a forest of stars of size at most 4. By slightly modifying our reduction, we further obtain NP-hardness for grid graphs (which are planar and have maximum degree $4$) with bounded pathwidth. Notice that, in a sense, these results clearly delineate the frontier of polynomial-time versus NP-complete cases for most standard graph parameters: On the one hand, when we place no bound on the degree, \SO\ is already NP-hard on graphs which are 2 vertices away from an extremely restricted class. On the other hand, for bounded-degree graphs no such result can be obtained, as any graph that has bounded tree-depth and bounded degree actually has bounded size, so the hardness for bounded pathwidth graphs is essentially best possible.

The results above motivate us to consider more restrictive parameters. In this direction, we present an XP algorithm parameterized by vertex cover number, $\vc$, that runs in time $n^{\bO(\vc^2)}$. Even though this improves the situation compared to the paraNP-completeness for treewidth, this complexity is rather disappointing, raising the natural question whether the square in the exponent is necessary. Our main contribution for this parameter is to answer this question by proving that an $n^{o(\vc^2)}$ time algorithm would refute the ETH.
Along the way,  we also show that \textsc{Steiner Orientation} is fixed-parameter tractable when parameterized by distance to a clique, which can be seen as the dense analogue of vertex cover (a vertex cover of a graph is a clique modulator of its complement). 

Moving in another direction, we consider the number of undirected or directed edges as parameters. Here, we show that \textsc{Steiner Orientation} is fixed-parameter tractable (FPT) when parameterized by the number $|A|$ of directed edges. We observe that if $|A|=0$, the problem is solvable in polynomial time as a consequence of Robbins' theorem~\cite{HASSIN1989}. Thus, our result generalizes this classical polynomial-time case. Here, it is natural to consider the number $|E|$ of undirected edges. Since we have only two directions for each undirected edge, \textsc{Steiner Orientation} is clearly solvable in $2^{|E|}n^{\bO(1)}$ time. Indeed, this running time is essentially tight, as a careful observation of the standard NP-hardness reduction for \textsc{Steiner Orientation}~\cite{ArkinH02} yields a $(2-\varepsilon)^{|E|} n^{\bO(1)}$ lower bound under the SETH. 

Finally, we investigate the parameterized complexity with respect to the combined parameters $k + \tw$ and $k + \vc$. We show that \textsc{Steiner Orientation} is FPT when parameterized by $k + \tw$ by formulating it as an MSO$_2$ problem, which stands in contrast to the NP-hardness on series-parallel graphs. Furthermore, we give a polynomial kernel when parameterized by $k + \vc$. 




\subparagraph*{Overview of Techniques.}  Our NP-completeness proof for \SO\ on graphs of treewidth $2$ is based on a direct reduction from a variant of \textsc{3-SAT}, where we naturally use undirected edges to represent variables of the initial formula, with directions representing truth assignments (this is standard). The key new insight is the following: we can transmit information between copies of such edges, ensuring they must be oriented in a consistent way, by adding demand pairs which must be routed through \emph{two} main hub vertices. By using appropriately oriented arcs to connect our gadgets to the two hubs we ensure that gadgets don't interfere with each other. In a sense, the take-home message of this construction is that the structure of $G$ alone is not enough to measure the complexity of an instance, because the way that demand pairs interact with the graph can significantly complicate the structure of the instance.

The XP algorithm parameterized by vertex cover is based on a simple branching procedure: for each pair $u,v$ of vertices of the vertex cover, guess if there is a directed path of length at most $2$ from $u$ to $v$ in the solution and, if yes, which is (potentially) the vertex $w$ in the middle of the path. This clearly leads to at most $n^{\bO(\vc^2)}$ guesses, and once we have fixed these decisions it is not hard to complete the solution because we can infer exactly which vertices of the cover have directed paths to each other.  

Given the simplicity of the algorithm sketched above for vertex cover, it is somewhat surprising that this is essentially optimal (under the ETH). We establish this by giving a reduction from $k$-\textsc{Clique} producing a \SO\ instance with vertex cover $\bO(\sqrt{k})$. The key intuition is again that we can use a small number ($\bO(\sqrt{k})$) of hubs, through which information will be routed. In particular, we set up a complete bipartite graph $K_{\sqrt{k},\sqrt{k}}$ and for each edge of this graph we set up a gadget encoding a selection in the original instance: the number of edges of this graph is sufficient to encode all selections, but the vertex cover of the construction is determined by the number of vertices of the complete bipartite graph. 

To obtain an algorithm for parameter $|A|$ we face significantly more obstacles than for parameter $|E|$ (which is trivially FPT). Our approach relies on exhaustive applications of some (standard) reduction rules for this problem, which eliminate all cycles and degree $1$ vertices. We then focus on a restricted special case of \SO, where all undirected components are paths, all internal vertices of the paths are incident on no arcs, and the endpoints of each path are incident on at most one arc. We solve this restricted case using a reduction to \textsc{2-SAT}. We then show, using arguments which essentially take into account that $A$ is a feedback edge set of the underlying graph, that a simple branching algorithm can reduce any instance into $2^{\bO(|A|)}$ instances of the (polynomial-time solvable) restricted case.

The FPT algorithm for distance to clique relies again on exhaustively applying reduction rules that eliminate cycles and then making some simple structural observations: there are few edges incident on the modulator vertices (so we can guess their orientation); and the edges contained in the clique must have a very simple structure (they form a matching). The FPT algorithm for $\tw+k$ relies on a formulation of the problem in MSO$_2$ logic and Courcelle's theorem -- we leave it as an interesting open problem to determine the best dependence for this parameterization. Finally, for the more restrictive parameter $\vc+k$ we obtain a polynomial kernel via a maximum matching argument. In particular, we show that by calculating a maximum matching between pairs of vertices of the vertex cover and non-terminals in the indpendent set, we can identify a set of $\bO(\vc^2)$ vertices of the independent set which are sufficient to obtain a solution and delete the rest, giving a kernel of order $\bO(\vc^2+k)$.

\subsection{Related work}
The \textsc{Steiner Orientation} problem has been studied extensively.
From a bioinformatics perspective, an optimization version of \textsc{Steiner Orientation} is widely studied~\cite{MedvedovskyBZS08,jcb/SilverbushES11,DornHKNU11,netmahib/Roayaei19}. In this optimization version, called \textsc{Maximum Steiner Orientation}, the goal is to find an orientation that maximizes the number of satisfied terminals pairs. 

Unfortunately, \textsc{Maximum Steiner Orientation} is significantly more difficult than \textsc{Steiner Orientation}. Medvedovsky et al.~\cite{MedvedovskyBZS08} show that the problem is inapproximable within a factor of 12/13 even on undirected stars and binary trees. Elberfeld et al. then observe that the NP-hardness reduction for \textsc{Steiner Orientation} in \cite{ArkinH02} implies the inapproximability within a factor of 7/8~\cite{tcs/ElberfeldSDSS13}. More recently, Hörsch \cite{Horsch25} considered the maximization version of \SO\ but under the restriction that the set of demand pairs includes all possible demands and showed that this case is also APX-hard. Interestingly, for this version of the problem Hörsch supplies an XP algorithm for parameter $|A|$ and leaves it as an open question whether an FPT algorithm can be obtained. In this paper we provide an FPT algorithm for \SO\ parameterized by $|A|$, but this is somewhat orthogonal to Hörsch's question, as we are not dealing with the maximization version (which makes our problem easier) but we are also not assuming that all pairs of vertices have a demand (which makes our problem harder).

For the approximability, an $\bO(\log n)$-approximation is presented for undirected graphs~\cite{MedvedovskyBZS08}. This has been improved to a factor $\bO(\log n/\log \log n)$~\cite{GamzuM16}, and later to $\bO(\log k/\log \log k)$~\cite{CyganKN13}. 
For bounded feedback vertex number graphs and bounded treewidth graphs, 
an $\bO(\log n)$-approximation algorithm and an $\bO(\log^2 n)$ approximation algorithm are presented~\cite{tcs/ElberfeldSDSS13}.

From the viewpoint of parameterized complexity, Dorn et al.~\cite{DornHKNU11} and  Roayaei~\cite{netmahib/Roayaei19} propose several FPT algorithms for parameters related to terminal pairs.

Regarding the parameterized approximation, W{\l}odarczyk~\cite{icalp/Wlodarczyk20} shows that there is no $(\log k)^{o(1)}$-approximation algorithm for \textsc{Maximum Steiner Orientation}  that runs in FPT time with respect to $k$ assuming FPT $\neq$ W[1].

\section{Preliminaries}\label{sec:prelim}
We use the standard notations in graph theory.
Let $G=(V,E,A)$ be a mixed graph where $V$ is a set of vertices, $E$ is a set of undirected edges, and $A$ is a set of directed edges. We denote by $\{u,v\}\in E$ an undirected edge between $u$ and $v$ and by $(u,v)\in A$ a directed edge from $u$ to $v$.

A mixed graph $G$ is a \emph{mixed acyclic graph} if it has no cycle. Notice that in the context of mixed graphs a cycle is any sequence of vertices $v_1,v_2,\ldots,v_p$ with $v_1=v_p$ such that between any two consecutive vertices $v_i,v_{i+1}$ the graph contains either the edge $\{v_i, v_{i+1}\}$ or the arc $(v_i,v_{i+1})$. 
By definition, the subgraph induced by undirected edges in a mixed acyclic graph is a forest. Moreover, the graph obtained from a mixed acyclic graph $G$ by contracting all the undirected edges is a directed acyclic graph.
For a vertex subset $V'\subseteq V$, we denote by $G[V'] = (V',E(V'),A(V'))$ the subgraph induced by $V'$,  where $E(V')\subseteq E$ and $A(V')\subseteq A$. A mixed path of a mixed graph is a path using both edges and arcs respecting the orientation of the arcs.
For a mixed graph $G = (V, E, A)$, $\lambda(E)$ denotes an orientation of $E$ and $G_\lambda$ denotes the directed graph obtained by the orientation $\lambda(E)$.

\subparagraph*{Graph parameters.} Throughout the paper we will use several structural graph parameters which will in general refer to the underlying graph of a given mixed graph $G$. Recall that the underlying graph is the graph obtained by replacing each arc by an edge with the same endpoints. The parameters we will mention are treewidth ($\tw$), pathwidth ($\pw$), vertex integrity ($\vi$), feedback vertex set ($\fvs$), vertex cover ($\vc$), and distance to clique ($\dtc$). For the definitions of treewidth and pathwidth we refer the reader to \cite{books/CyganFKLMPPS15}; $\fvs$, $\vc$, and $\dtc$ denote the size of the smallest set of vertices whose removal leaves the graph a forest, an independent set, or a clique respectively; while a graph has vertex integrity at most $k$ if there exists a set of vertices $S$ whose removal results in a graph where all components have size at most $k-|S|$. It is known that four of these parameters form a hierarchy, in the sense that $\tw \le \pw \le \vi \le \vc+1$ for all graphs. We note that all these parameters are closed under vertex deletion and edge contraction, meaning that these operations can only decrease the parameter value of a graph. 

\subsection{Preprocessing}

We present two basic polynomial-time reduction rules for \SO, which eliminate cycles and degree-$1$-vertices respectively. Applying these rules will never increase the values of our parameters, as the rules use edge contractions and vertex deletions, so in the rest of the paper we will mostly focus on instances where these rules have been exhaustively applied.

\begin{proposition}[\cite{jcb/SilverbushES11}]\label{prop:prepoc}
    Let $(G,\terminals)$ be an instance of \SO.
    Let $G'$ be a mixed graph obtained from $G$ by contracting all the cycles and $\terminals'$
    be the corresponding set of terminal pairs.
    Then $(G,\terminals)$ is a yes-instance if and only if $(G',\terminals')$ is a yes-instance.
    Moreover, $G'$ is a mixed acyclic graph and can be computed in polynomial time.
\end{proposition}


\begin{propositionrep}\label{prop:remove:degree:1}
    Let $(G,\terminals)$ be an instance of \SO.
    If $G$ has a vertex $v\in V$ of degree $1$ in the underlying graph, then we can either correctly conclude that $(G,\terminals)$ is a no-instance or construct a new instance $(G',\terminals')$ such that $G'=G[V\setminus \{v\}]$ and $(G,\terminals)$ is a yes-instance if and only if $(G',\terminals')$ is a yes-instance.
\end{propositionrep}

\begin{proof}
First, we assume that there are no terminal pairs $(s,t) \in \terminals$ such that $s = t$. Indeed, any such pair is always connected, so we can safely remove it from $\terminals$.

Let $v \in V$ be a vertex of degree $1$. If there is no terminal pair $(s,t) \in \terminals$ such that $s = v$ or $t = v$, then we can safely delete $v$ since it will never be used in any $s$-$t$ path.

Now suppose there exists a terminal pair $(s,t) \in \terminals$ where either $v = s$ or $v = t$, and let $u \in V$ be the unique neighbor of $v$ in $G$.

If $v$ appears as both a source terminal and a sink terminal, then we have a No-instance (as $v$ has degree one). Therefore, $v$ appears only as a source or only as a sink. If the edge connecting $u$ to $v$ is directed and has an incorrect orientation (that is, $v$ is a source but the arc is going into $v$, or $v$ is a sink but the arc is coming out of $v$), then we have a No-instance. 

If we have not concluded that we have a No-instance we can always go from $v$ to $u$ if $v$ is a source terminal and from $u$ to $v$ if $v$ is a sink terminal. Since, $v$ is of degree one we also have that:
\begin{itemize}
    \item If $v = s$, then any $v$-$t$ path must have the form $v$-$u$-$\ldots$-$t$.
    \item If $v = t$, then any $s$-$v$ path must have the form $s$-$\ldots$-$u$-$v$.
\end{itemize}
Therefore, we can guarantee the existence of a $v$-$t$ path by providing a $u$-$t$ path, or guarantee the existence of an $s$-$v$ path by providing an $s$-$u$ path. Since $v$ has degree one, the reverse is also true. By replacing all appearances of $v$ with $u$ in any terminal pairs, we create an equivalent instance $(G,\terminals')$ of \SO. 

Finally, since in $(G,\terminals')$ the vertex $v$ does not appear in any terminal pairs and has degree one, it can be removed. Thus, $(G-\{v\},\terminals')$ is a yes-instance of \SO if and only if $(G,\terminals)$ is a yes-instance of \SO.
\end{proof}


\section{NP-hardness}\label{sec:np_hardness}

Our main result in this section is to show that {\SO} remains NP-hard even on very restricted families of graphs. We focus on restrictions on the structure of the underlying undirected graph and show that even if the underlying graph is series-parallel and has vertex-deletion distance $2$ to a forest of stars of size at most $4$ the problem is still NP-complete.
Note that all series-parallel graphs have treewidth at most $2$~\cite[Theorem~41]{tcs/Bodlaender98} and are planar~\cite[Chapter~11.2]{books/BrandstadtLS99},
while the graph constructed by our reduction has feedback vertex number~$2$, pathwidth~$3$, and vertex integrity~$6$.
Since {\SO} is polynomial-time solvable when the underlying graph of the input is a tree (as there exists a unique path between every pair of terminals),
our result establishes a sharp dichotomy on the polynomial-time solvability of {\SO} with respect to the treewidth of the input graph.
Moreover, by slightly modifying our reduction we further obtain NP-hardness for grid graphs of bounded pathwidth.

\begin{theorem}\label{thm:np_hardness:main}
    {\SO} is NP-complete on series-parallel graphs of vertex-deletion distance $2$ to a forest of stars of size at most $4$.
\end{theorem}

\begin{proof}
    It is easy to see that {\SO} belongs to NP, thus in the rest of the proof we argue about its NP-hardness.
    The starting point of our reduction is the \textsc{Monotone 3-SAT} problem,
    which is the variation of 3-SAT where every clause contains at most 3 literals and is \emph{monotone}, that is,
    it contains only unnegated or only negated variables.
    This variation is well-known to be NP-hard~\cite{dam/DarmannD21,iandc/Gold78}.
    
    Let $\phi$ be an instance of \textsc{Monotone 3-SAT}, where $X = \{ x_1, \ldots, x_n \}$ denotes its variables and
    $C = \{ c_1, \ldots, c_m \}$ its clauses.
    We say that a clause is \emph{positive} if it consists of only unnegated variables,
    and \emph{negative} if it consists of only negated variables.
    
    We will construct an instance $(G, \terminals)$ of {\SO} that is equivalent to~$\phi$.
    For each variable $x_i$, introduce vertices $\ell_i,r_i$ that are connected via an undirected edge.
    We further introduce vertices $\ell$ and $r$, as well as the arcs $\ell_i \to \ell$ and $r_i \to r$, for all $i \in [n]$.
    Now let $c_j$ be a positive clause containing 3 variables, that is,
    $c_j = x_{i_1} \lor x_{i_2} \lor x_{i_3}$ for distinct $i_1,i_2,i_3 \in [n]$.
    We construct the corresponding clause gadget as depicted in \cref{fig:np_hardness}.
    In particular, we introduce vertices $v^j_{i_1}, v^j_{i_2}, v^j_{i_3}$, all of which have an incoming arc from $\ell$,
    as well as a vertex $t_j$ that is connected via undirected edges with $v^j_{i_1}, v^j_{i_2}, v^j_{i_3}$ and has an incoming arc from $r$.
    We further introduce the terminal pairs $(\ell,t_j), (r_{i_1},v^j_{i_1}), (r_{i_2},v^j_{i_2}), (r_{i_3},v^j_{i_3})$ into $\terminals$.
    We treat positive clauses that contain 2 variables in an analogous way, 
    while the construction for the negative clauses is entirely symmetric.

    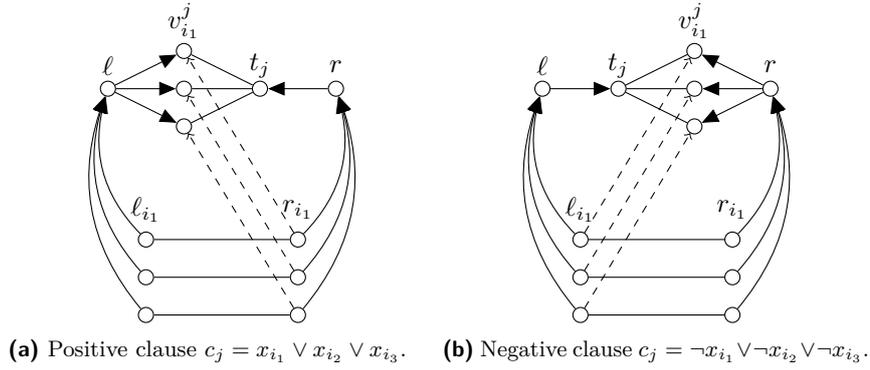
\begin{figure}[ht]
        \centering
            \begin{subfigure}[b]{0.4\linewidth}
            \centering
            \begin{tikzpicture}[transform shape]
        
            \node[vertex] (a) at (5,5) {};
            \node[] () at (5,5.3) {$\ell$};
            
            \node[vertex] (x1) at (6,4.5) {};
            \node[vertex] (x2) at (6,5) {};
            \node[vertex] (x3) at (6,5.5) {};
            \node[] () at (6,5.9) {$v^j_{i_1}$};

            \node[vertex] (t) at (7,5) {};
            \node[] () at (7,5.3) {$t_j$};
            
            \node[vertex] (b) at (8,5) {};
            \node[] () at (8,5.3) {$r$};
            
            \node[vertex] (x11) at (5.5,3) {};
            \node[] () at (5.5,3.4) {$\ell_{i_1}$};
            \node[vertex] (x12) at (7.5,3) {};
            \node[] () at (7.5,3.4) {$r_{i_1}$};
            
            \node[vertex] (x21) at (5.5,2.5) {};
            \node[vertex] (x22) at (7.5,2.5) {};
            \node[vertex] (x31) at (5.5,2) {};
            \node[vertex] (x32) at (7.5,2) {};
            
            \draw[] (x11)--(x12);
            \draw[] (x21)--(x22);
            \draw[] (x31)--(x32);
            
            \draw[-triangle 45] (x11) edge [bend left] (a);
            \draw[-triangle 45] (x21) edge [bend left] (a);
            \draw[-triangle 45] (x31) edge [bend left] (a);
            
            \draw[-triangle 45] (x12) edge [bend right] (b);
            \draw[-triangle 45] (x22) edge [bend right] (b);
            \draw[-triangle 45] (x32) edge [bend right] (b);

            \draw[-triangle 45] (a)--(x1);
            \draw[-triangle 45] (a)--(x2);
            \draw[-triangle 45] (a)--(x3);
            
            \draw[] (x1)--(t);
            \draw[] (x2)--(t);
            \draw[] (x3)--(t);
            
            \draw[-triangle 45] (b)--(t);
    
            \draw[dashed, ->] (x12)--(x3);
            \draw[dashed, ->] (x22)--(x2);
            \draw[dashed, ->] (x32)--(x1);
    
            \end{tikzpicture}
            \caption{Positive clause $c_j = x_{i_1} \lor x_{i_2} \lor x_{i_3}$.}
        \end{subfigure}
        \begin{subfigure}[b]{0.4\linewidth}
        \centering
            \begin{tikzpicture}[transform shape]
        
            \node[vertex] (a) at (5,5) {};
            \node[] () at (5,5.3) {$\ell$};

            \node[vertex] (t) at (6,5) {};
            \node[] () at (6,5.3) {$t_j$};
            
            \node[vertex] (x1) at (7,4.5) {};
            \node[vertex] (x2) at (7,5) {};
            \node[vertex] (x3) at (7,5.5) {};
            \node[] () at (7,5.9) {$v^j_{i_1}$};
            
            \node[vertex] (b) at (8,5) {};
            \node[] () at (8,5.3) {$r$};
            
            \node[vertex] (x11) at (5.5,3) {};
            \node[] () at (5.5,3.4) {$\ell_{i_1}$};
            \node[vertex] (x12) at (7.5,3) {};
            \node[] () at (7.5,3.4) {$r_{i_1}$};
            
            \node[vertex] (x21) at (5.5,2.5) {};
            \node[vertex] (x22) at (7.5,2.5) {};
            \node[vertex] (x31) at (5.5,2) {};
            \node[vertex] (x32) at (7.5,2) {};
            
            \draw[] (x11)--(x12);
            \draw[] (x21)--(x22);
            \draw[] (x31)--(x32);
            
            \draw[-triangle 45] (x11) edge [bend left] (a);
            \draw[-triangle 45] (x21) edge [bend left] (a);
            \draw[-triangle 45] (x31) edge [bend left] (a);
            
            \draw[-triangle 45] (x12) edge [bend right] (b);
            \draw[-triangle 45] (x22) edge [bend right] (b);
            \draw[-triangle 45] (x32) edge [bend right] (b);

            \draw[-triangle 45] (b)--(x1);
            \draw[-triangle 45] (b)--(x2);
            \draw[-triangle 45] (b)--(x3);
            
            \draw[] (x1)--(t);
            \draw[] (x2)--(t);
            \draw[] (x3)--(t);
            
            \draw[-triangle 45] (a)--(t);
    
            \draw[dashed, ->] (x11)--(x3);
            \draw[dashed, ->] (x21)--(x2);
            \draw[dashed, ->] (x31)--(x1);
    
            \end{tikzpicture}
            \caption{Negative clause $c_j = \neg x_{i_1} \lor \neg x_{i_2} \lor \neg x_{i_3}$.}
        \end{subfigure}
    \caption{The terminal pairs added during the construction of the clause gadget are denoted by dashed arcs,
    apart from $(\ell, t_j)$ and $(r, t_j)$ respectively for the sake of clarity.}
    \label{fig:np_hardness}
    \end{figure}    

    This completes the construction of the instance $(G,\terminals)$ of \SO.
    It is easy to see that $G$ is series-parallel,
    while $G - \{\ell,r\}$ is a forest whose connected components are either a $K_2$, a $K_{1,2}$, or a $K_{1,3}$,
    thus indeed it holds that $\fvs(G) = 2$, $\pw(G) = 3$, and $\vi(G) = 6$.
    It remains to argue that the instance is equivalent to $\phi$.

    For the forward direction, let $\alpha \colon X \to \{ \true, \false \}$ be a satisfying assignment for $\phi$.
    Consider the following orientation for the undirected edges of $G$.
    If $\alpha(x_i) = \true$ then we orient the edge $\{ \ell_i , r_i \}$ as $r_i \to \ell_i$,
    otherwise if $\alpha(x_i) = \false$ then we orient it as $\ell_i \to r_i$.
    As for the clauses $c_j$ that contain $x_i$ and the edge $\{ v^j_i, t_j \}$,
    if $c_j$ is a positive clause (resp.~negative) then
    if $\alpha(x_i) = \true$ we orient it as $v^j_i \to t_j$ (resp.~$t_j \to v^j_i$),
    otherwise if $\alpha(x_i) = \false$ we orient it as $t_j \to v^j_i$ (resp.~$v^j_i \to t_j$).
    It suffices to argue that this orientation satisfies all terminal pairs in $\terminals$.

    Let $c_j = x_{i_1} \lor x_{i_2} \lor x_{i_3}$ be a positive clause involving three variables.
    We will argue that the described orientation satisfies all terminal pairs introduced due to the clause gadget of $c_j$.
    Since $\alpha$ is a satisfying assignment for $\phi$, it holds that there exists $i \in \{ i_1,i_2,i_3\}$ such that $\alpha(x_i) = \true$.
    Consequently, the arc $v^j_i \to t_j$ exists in our orientation, thus the terminal pair $(\ell,t_j)$ is satisfied by the directed path
    $\ell \to v^j_i \to t_j$.
    As for the terminal pair $(r_{i'},v^j_{i'})$, where $i' \in \{ i_1,i_2,i_3\}$,
    notice that if $\alpha(x_{i'}) = \true$ then it is satisfied by the directed path $r_{i'} \to \ell_{i'} \to \ell \to v^j_{i'}$,
    while if $\alpha(x_{i'}) = \false$ then it is satisfied by the directed path $r_{i'} \to r \to t_j \to v^j_{i'}$.
    The argumentation is similar for positive clauses involving two variables, and symmetric for negative clauses.

    For the converse direction, assume there exists an orientation of $G$ such that all terminal pairs in $\terminals$ are satisfied.
    Consider the assignment $\alpha \colon X \to \{ \true, \false \}$ where $\alpha(x_i) = \true$ if and only if the edge $\{ \ell_i, r_i \}$
    has been oriented as $r_i \to \ell_i$.
    We will prove by contradiction that this is a satisfying assignment for $\phi$.
    Let $c_j = x_{i_1} \lor x_{i_2} \lor x_{i_3}$ be a positive clause involving three variables,
    and assume that $\alpha(x_{i}) = \false$ for all $i \in \{ i_1,i_2,i_3 \}$.
    In that case, for the terminal pair $(r_i, v^j_i)$ to be satisfied, it follows that the edge $\{ v^j_i, t_j \}$ has been directed as $t_j \to v^j_i$.
    Since this holds for all $i \in \{ i_1,i_2,i_3 \}$, it contradicts the fact that the terminal pair $(\ell,t_j)$ is satisfied by the orientation.
    The argumentation is similar for positive clauses involving two variables, and symmetric for negative clauses.
\end{proof}


\begin{corollaryrep}\label{cor:grid}
    {\SO} is NP-complete on grid graphs of constant pathwidth.
\end{corollaryrep}
\begin{proof}
     (Sketch) As shown in \Cref{fig:grid} in the appendix, 
     We substitute the clause gadgets in the reduction of \cref{thm:np_hardness:main} by the alternative one with vertices $\ell'_j$ and $r'_j$ corresponding to $\ell$ and $r$  and connecting $\ell_1,\ldots, \ell_n, \ell'_1, \ldots, \ell'_m$ (resp., $r_1, \ldots, r_n, r'_1,  \ldots, r'_m$) by 
long directed paths.
It can be argued that the instance is equivalent to $\phi$ as in \cref{thm:np_hardness:main}.
\begin{figure}
    \centering
    \includegraphics[width=0.3\linewidth]{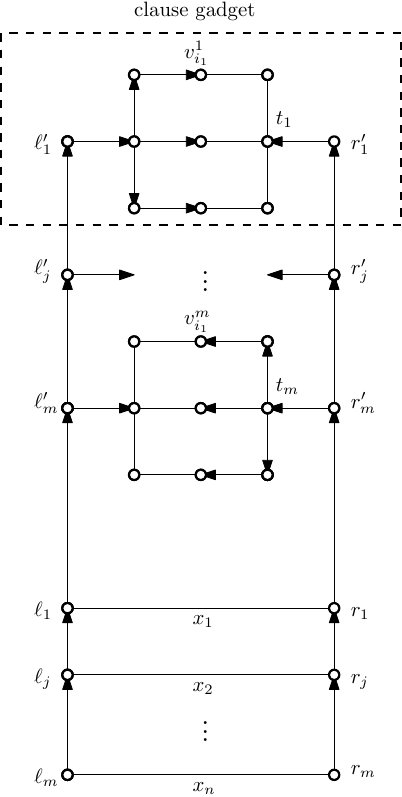}
    \caption{The constructed grid graph in \Cref{cor:grid}.}
    \label{fig:grid}
\end{figure}
\end{proof}

\section{FPT algorithm by distance to clique}\label{sec:dtc}

In this section, we give an FPT algorithm parameterized by distance to clique.
More precisely, we consider the case where the underlying undirected graph of
the input is at most $\dtc$ vertex deletions away from being a clique. In a
sense, this parameterization is the dense analogue of the parameterization by
vertex cover (because a vertex cover of a graph is a clique modulator of its
complement). As usual, we will assume that the set whose removal leaves the
underlying graph a clique is given to us in the input. The main result of this
section is the following.

\begin{theoremrep}\label{thm:dtc} {\SO} can be solved in time $4^\dtc n^{\bO(1)}$.
\end{theoremrep}

\begin{proof}[Proof of \cref{thm:dtc}]

We observe that the number of undirected edges in $E(S)$ is at most $\dtc$ from
the fact that $G$ is mixed acyclic and the number of undirected edges between
$S$ and $V\setminus S$ is at most $\dtc$ from \Cref{lem:undirectededge:dtc}.
Thus, we enumerate the at most $2^{2\cdot \dtc}=4^\dtc$ orientations of these
undirected edges and solve each resulting instance. In the remainder suppose
that all undirected edges have both endpoints in $V\setminus S$ and recall that
such edges form a matching by \cref{lem:undirectededge:clique}. 

Topologically sort the vertices of $V\setminus S$ so that whenever $(u,v)$ is
an arc we have that $u$ is sorted before $v$. This produces a partial ordering
$\sigma$ such that two vertices are only incomparable if they are the endpoints
of an undirected edge.

%
%
%
%
%

Let us consider a terminal pair $(s,t)$. If a directed path from $s$ to $t$
exists without using any edges of $E$, we simply remove this terminal pair, as it is trivially satisfied. On
the other hand, if no such path exists in the graph obtained by replacing each
edge of $E$ with a pair of anti-parallel arcs, we have a No instance and we can
reject. If neither case applies, we claim that either we can remove the pair
$(s,t)$ or there exists an edge $e$ such that orienting it in one direction
makes it impossible to connect $s$ to $t$ (and hence we are forced to choose
the other orientation).  Furthermore, we can find this edge in polynomial time,
allowing us to repeatedly simplify the instance until all pairs are satisfied
or we conclude we have a No instance.

Consider the directed paths from $s$ to $t$ in the graph where we have replaced
edges with anti-parallel arcs. We claim that there is a special edge $e$ which
appears in all these paths. To see this, take two (not necessarily distinct)
paths $P,P'$ from $s$ to $t$, both of which must use some undirected edges.
Suppose without loss of generality that the first edge of $P$ (call it $e$)
comes before the last edge of $P'$ (call it $e'$) in the ordering $\sigma$. If
$e\neq e'$ we have a contradiction: there is a directed path from $s$ to an
endpoint of $e$ (as $e$ is the first undirected edge of $P$); there are arcs
from $e$ to $e'$ from the ordering $\sigma$; and there is a directed path from
$e'$ to $t$, as $e'$ is the last undirected edge of $P'$.

We therefore conclude that $e=e'$. Because the paths $P,P'$ were not
necessarily distinct, applying this argument to $P$ with itself implies that
every path uses exactly one undirected edge. But then, any two paths $P,P'$
must in fact use the same unique undirected edge $e$. Furthermore, this edge
can be found in polynomial time via a shortest path algorithm which attempts to
find a path using as few undirected edges as possible.

Consider now the instances we obtain by orienting $e$ in either direction. If
both instances contain a directed path from $s$ to $t$ we can remove the pair
$(s,t)$ from the input, as it is always satisfied. Otherwise, one orientation
of $e$ is forbidden and we orient $e$ in the other direction. This satisfies
the terminal pair and we proceed in the same way until the whole instance has
been processed.  \end{proof}

We establish \cref{thm:dtc} starting via a series of reduction rules given by
the lemmas below. Throughout we will assume that we are given a set $S$ of
$\dtc$ vertices such that removing $S$ from the input graph leaves a mixed
graph whose underlying graph is a clique.  We will also assume that, using
\cref{prop:prepoc}, the input is a mixed \emph{acyclic} graph.

\begin{lemmarep}\label{lem:undirectededge:dtc} Let $G$ be a mixed acyclic graph
and $S$ a set of vertices such that the underlying graph of $G-S$ is a clique.
Then, each vertex in $S$ has at most one undirected edge to $V\setminus S$.
\end{lemmarep}

\begin{proof} 
Suppose that $u\in S$ has two undirected edges $\{u,v\}$ and
$\{u,w\}$ for $v,w,\in V\setminus S$. The underlying graph of $G[V\setminus S]$
forms a clique, so there exists either $(v,w)$, $(w,v)$, or $\{v,w\}$.
Therefore, $u,v,w$ form a cycle, which contradicts that $G$ is a mixed acyclic graph.  
\end{proof}

\begin{lemmarep}\label{lem:undirectededge:clique}
    In a mixed acyclic graph $G$, let $C$ be a clique in the underlying graph of $G$. Then there is no vertex incident to at least two undirected edges in $C$. In other words, the set of undirected edges in $C$ forms a matching.
\end{lemmarep}
\begin{proof}
    The proof is similar to \Cref{lem:undirectededge:dtc}.
    Suppose that there exists a vertex $u\in V(C)$ incident to undirected edges $\{u,v\},\{u,w\}$ for $v,w\in V(C)$. Since $C$ is a clique, there exists either $(v,w)$, $(w,v)$, or $\{v,w\}$. Therefore, $u,v,w$ form a cycle, which contradicts that $G$ is a mixed acyclic graph. 
\end{proof}

%
%

We can now proceed to the proof of \cref{thm:dtc}. The high-level idea is that
by using our previous observations the number of undirected edges incident on
$S$ is at most $2\dtc$; we will therefore enumerate all possible orientations
of these edges and check for each orientation if it can be extended to a
feasible solution. The challenge will be in establishing that this check can be
performed in polynomial time.

\section{Vertex cover number}\label{sec:vc}

In this section we consider {\SO} parameterized by the vertex cover number $\vc$ of the input graph,
arguably one of the most restrictive structural parameterizations one could consider.
We first show in \cref{thm:w1_hardness} that even in this extremely restrictive setting,
the problem remains intractable and prove its W[1]-hardness.
As a matter of fact, our reduction implies that the problem does not admit any $n^{o(\vc^2)}$-time algorithm under the ETH,
and our next result is in \cref{thm:algo_vc} to present such an optimal algorithm.

\subsection{Hardness}

\begin{theorem}\label{thm:w1_hardness}
    {\SO} is W[1]-hard parameterized by the vertex cover number $\vc$ of the input graph,
    and for any computable function $f$ it cannot be solved in time $f(\vc) n^{o(\vc^2)}$ under the ETH.
\end{theorem}

\begin{proof}
    In {\kMC} we are given a graph $G$ and a partition of $V(G)$ into $k$ independent sets each of size $n$,
    and we are asked to determine whether $G$ contains a $k$-clique.
    It is well-known that {\kMC} parameterized by $k$ is W[1]-hard and does not admit any $f(k) n^{o(k)}$-time algorithm,
    where $f$ is any computable function,
    unless the ETH is false~\cite{books/CyganFKLMPPS15}.
    Let $(G, k)$ be an instance of \kMC, where we assume without loss of generality that $\sqrt{k} \in \mathbb{N}$
    (one can do so by adding dummy independent sets connected to all the other vertices of the graph).
    Recall that we assume that $G$ is given to us partitioned into $k$ independent sets $V_1, \ldots V_k$,
    where $V_i = \{ v^i_1, \ldots, v^i_n \}$.
    We will construct in polynomial time an equivalent instance $(H,\terminals)$ of \SO, with $\vc(H) = \bO( \sqrt{k} )$.

    \proofsubparagraph{Construction.}
    We first introduce the vertices $\setdef{\ell_\alpha, \, \ell'_\alpha, \, r_\alpha, \, r'_\alpha}{\alpha \in [ \sqrt{k} ]}$
    and add the arcs $\ell_{\alpha} \to \ell'_{\beta}$ and $r_{\alpha} \to r'_{\beta}$ for all $\alpha , \beta \in [ \sqrt{k} ]$.
    For all $i \in [k]$, we further introduce the independent sets $X_i = \{x^i_1, \ldots, x^i_n \}$
    and $Y_i = \{ y^i_1, \ldots, y^i_n \}$.
    We fix a bijective function $h \colon [k] \to [ \sqrt{k} ]^2$ that maps
    every independent set $V_i$ to a distinct pair $(\alpha, \beta) \in [ \sqrt{k} ]^2$,
    and for all $i \in [k]$ with $h(i) = (\alpha, \beta)$
    \begin{itemize}
        \item we add the edge $\{ \ell_{\alpha}, x^i_j \}$ for all $x^i_j \in X_i$,
        \item we add the arc $(x^i_j, r_{\beta})$ for all $x^i_j \in X_i$,
        \item we add the arc $(\ell'_{\alpha}, y^i_j)$ for all $y^i_j \in Y_i$,
        \item we add the edge $\{ r'_{\beta}, y^i_j \}$ for all $y^i_j \in Y_i$.
    \end{itemize}
    This completes the construction of the graph $H$,
    and we refer to \cref{fig:w1_hardness} for an illustration of a part of it.
    As for the terminal pairs, we proceed in three steps:
    \begin{enumerate}
        \item We initially add into $\terminals$ the pairs $(\ell_{\alpha}, r_{\beta})$ and $(\ell'_{\alpha}, r'_{\beta})$,
        for all $\alpha, \beta \in [  \sqrt{k} ]$.

        \item Next, for all $i \in [k]$ and $j \in [n]$,
        we add the terminal pair $(x^i_j, y^i_{j'})$ into $\terminals$, for all $j' \in [n] \setminus \{j\}$.
        We refer to the terminal pairs added in this step as \emph{consistency terminal pairs}.

        \item Lastly, for distinct $i_1, i_2 \in [k]$,
        if $\{ v^{i_1}_{j_1}, v^{i_2}_{j_2} \} \notin E(G)$
        we add the terminal pairs $( x^{i_1}_{j_1}, y^{i_2}_{j_2}), (x^{i_2}_{j_2}, y^{i_1}_{j_1})$ into $\terminals$, where $j_1, j_2 \in [n]$.
        We refer to the terminal pairs added in this step as \emph{edge-checking terminal pairs}.

    \end{enumerate}
    This completes the construction of the instance $(H,\terminals)$.

    \begin{figure}[t]
    \centering
    \begin{tikzpicture}[scale=0.6, transform shape]
    
    \node[vertex] (v11) at (5,15) {};
    \node[] () at (5,14.6) {$\vdots$};
    \node[vertex] (v12) at (5,14) {};
    \draw[] (4.75,13.75) rectangle (5.25,15.25);

    \begin{scope}[shift={(0,-2)}]
        \node[vertex] (v21) at (5,15) {};
        \node[] () at (5,14.6) {$\vdots$};
        \node[vertex] (v22) at (5,14) {};
        \draw[] (4.75,13.75) rectangle (5.25,15.25);
    \end{scope}

    \begin{scope}[shift={(0,-3.25)}]
        \node[] () at (5,14.6) {$\vdots$};
    \end{scope}

    \begin{scope}[shift={(0,-4.5)}]
        \node[vertex] (vk1) at (5,15) {};
        \node[] () at (5,14.6) {$\vdots$};
        \node[vertex] (vk2) at (5,14) {};
        \draw[] (4.75,13.75) rectangle (5.25,15.25);
    \end{scope}

    \node[vertex] (l11) at (4,14.5) {};
    \node[] () at (4,14.8) {$\ell_1$};
    \node[vertex] (l1k) at (4,10) {};
    \node[] () at (4.1,10.4) {$\ell_{\sqrt{k}}$};

    \draw[] (l11)--(v11);
    \draw[] (l11)--(v12);
    \draw[] (l11) edge [bend right] (v21);
    \draw[] (l11) edge [bend right] (v22);
    \draw[] (l1k)--(vk1);
    \draw[] (l1k)--(vk2);

    \begin{scope}[shift={(2,0)}]
        \node[vertex] (r11) at (4,14.5) {};
        \node[] () at (4,14.8) {$r_1$};

        \node[vertex] (r12) at (4,12.5) {};
        \node[] () at (4,12.75) {$r_2$};
        
        \node[vertex] (r1k) at (4,10) {};
        \node[] () at (4.2,10.3) {$r_{\sqrt{k}}$};
    \end{scope}

    \draw[-triangle 45] (v11)--(r11);
    \draw[-triangle 45] (v12)--(r11);
    \draw[-triangle 45] (v21)--(r12);
    \draw[-triangle 45] (v22)--(r12);
    \draw[-triangle 45] (vk1)--(r1k);
    \draw[-triangle 45] (vk2)--(r1k);

    \begin{scope}[shift={(0,-8)}]
    
        \node[vertex] (u11) at (5,15) {};
        \node[] () at (5,14.6) {$\vdots$};
        \node[vertex] (u12) at (5,14) {};
        \draw[] (4.75,13.75) rectangle (5.25,15.25);

        \begin{scope}[shift={(0,-2)}]
            \node[vertex] (u21) at (5,15) {};
            \node[] () at (5,14.6) {$\vdots$};
            \node[vertex] (u22) at (5,14) {};
            \draw[] (4.75,13.75) rectangle (5.25,15.25);
        \end{scope}
    
        \begin{scope}[shift={(0,-3.25)}]
            \node[] () at (5,14.6) {$\vdots$};
        \end{scope}
    
        \begin{scope}[shift={(0,-4.5)}]
            \node[vertex] (uk1) at (5,15) {};
            \node[] () at (5,14.6) {$\vdots$};
            \node[vertex] (uk2) at (5,14) {};
            \draw[] (4.75,13.75) rectangle (5.25,15.25);
        \end{scope}

        \node[vertex] (l21) at (4,14.5) {};
        \node[] () at (3.7,14.3) {$\ell'_1$};
        \node[vertex] (l2k) at (4,10) {};
        \node[] () at (3.7,9.7) {$\ell'_{\sqrt{k}}$};
    
        \draw[-triangle 45] (l21)--(u11);
        \draw[-triangle 45] (l21)--(u12);
        \draw[-triangle 45] (l21) edge [bend right] (u21);
        \draw[-triangle 45] (l21) edge [bend right] (u22);
        \draw[-triangle 45] (l2k)--(uk1);
        \draw[-triangle 45] (l2k)--(uk2);
    
        \begin{scope}[shift={(2,0)}]
            \node[vertex] (r21) at (4,14.5) {};
            \node[] () at (4,14.2) {$r'_1$};
    
            \node[vertex] (r22) at (4,12.5) {};
            \node[] () at (4,12.2) {$r'_2$};
            
            \node[vertex] (r2k) at (4,10) {};
            \node[] () at (4.2,9.7) {$r'_{\sqrt{k}}$};
        \end{scope}
    
        \draw[] (u11)--(r21);
        \draw[] (u12)--(r21);
        \draw[] (u21)--(r22);
        \draw[] (u22)--(r22);
        \draw[] (uk1)--(r2k);
        \draw[] (uk2)--(r2k);
        
    \end{scope}

    \draw[-triangle 45] (l11) edge [bend right] (l21);
    \draw[-triangle 45] (l11) edge [bend right] (l2k);
    \draw[-triangle 45] (l1k) edge [bend right] (l21);
    \draw[-triangle 45] (l1k) edge [bend right] (l2k);

    \draw[-triangle 45] (r11) edge [bend left] (r21);
    \draw[-triangle 45] (r11) edge [bend left] (r22);
    \draw[-triangle 45] (r11) edge [bend left] (r2k);
    \draw[-triangle 45] (r12) edge [bend left] (r21);
    \draw[-triangle 45] (r12) edge [bend left] (r22);
    \draw[-triangle 45] (r12) edge [bend left] (r2k);
    \draw[-triangle 45] (r1k) edge [bend left] (r21);
    \draw[-triangle 45] (r1k) edge [bend left] (r22);
    \draw[-triangle 45] (r1k) edge [bend left] (r2k);    
    
    \end{tikzpicture}
    \caption{Part of the construction of graph $H$.
    Rectangles denote independent sets of size $n$.}
    \label{fig:w1_hardness}
    \end{figure}

    It is easy to see that $\setdef{\ell_\alpha, \, \ell'_\alpha, \, r_\alpha, \, r'_\alpha}{\alpha \in [ \sqrt{k} ]}$ is a vertex cover of $H$,
    thus $\vc(H) = \bO(\sqrt{k})$.
    To complete the proof, we show the following claim.

    \begin{claimrep}\label{claim:lMCtoSO}
            The instance $(G,k)$ of {\kMC} is a yes-instance if and only if the constructed instance $(H,\terminals)$ of \SO\ is a yes-instance.
    \end{claimrep}

    \begin{claimproof}
    For the forward direction, let $s \colon [k] \to [n]$ such that $\mathcal{V} = \setdef{v^i_{s(i)}}{i \in [k]} \subseteq V(G)$ is a $k$-clique of $G$.
    Let $i \in [k]$ where $h(i) = (\alpha, \beta)$.
    Consider the orientation where we orient the edges $\{ \ell_{\alpha}, x^i_j \}$ and $\{ r'_{\beta}, y^i_j \}$ as follows:
    \begin{itemize}
        \item if $j = s(i)$, then $\ell_{\alpha} \to x^i_j$ and $y^i_j \to r'_\beta$,
        \item if $j \neq s(i)$, then $x^i_j \to \ell_\alpha$ and $r'_\beta \to y^i_j$.
    \end{itemize}
    We argue that this orientation satisfies all terminal pairs.
    For the terminal pairs added in Step~1, that is indeed the case since $h$ is a bijection,
    i.e., for all $(\alpha,\beta) \in [\sqrt{k}]^2$ there exists $i \in [k]$ with $h(i) = (\alpha, \beta)$,
    thus the terminal pairs $(\ell_\alpha, r_\beta)$ and $(\ell'_\alpha, r'_\beta)$ are satisfied by the paths
    $\ell_\alpha \to x^i_{s(i)} \to r_\beta$ and $\ell'_\alpha \to y^i_{s(i)} \to r'_\beta$ respectively.
    For the consistency terminal pairs, fix $i \in [k]$ with $h(i) = (\alpha,\beta)$ and notice that all terminal pairs $(x^i_j, y^i_{j'})$ with $j \neq s(i)$
    are satisfied by the path $x^i_j \to \ell_\alpha \to \ell'_\alpha \to y^i_{j'}$.
    As for the case when $j = s(i)$, those terminal pairs are satisfied via the path $x^i_{s(i)} \to r_\beta \to r'_\beta \to y^i_{j'}$
    (recall that $j' \neq s(i)$).
    Lastly, for the edge-checking terminal pairs, consider a non-edge $\{ v^{i_1}_{j_1}, v^{i_2}_{j_2} \} \notin E(G)$,
    where $h(i_1) = (\alpha_1,\beta_1)$ and $h(i_2) = (\alpha_2,\beta_2)$ for distinct $i_1,i_2 \in [k]$.
    Since $\mathcal{V}$ is a clique, it holds that it contains at most one among $v^{i_1}_{j_1}$ and $v^{i_2}_{j_2}$.
    Assume without loss of generality that $v^{i_1}_{j_1} \notin \mathcal{V}$.
    In that case, the terminal pairs $(x^{i_1}_{j_1}, y^{i_2}_{j_2})$ and $(x^{i_2}_{j_2}, y^{i_1}_{j_1})$ are satisfied via the paths
    $x^{i_1}_{j_1} \to \ell_{\alpha_1} \to \ell'_{\alpha_2} \to y^{i_2}_{j_2}$
    and $x^{i_2}_{j_2} \to r_{\alpha_2} \to r'_{\alpha_1} \to y^{i_1}_{j_1}$
    respectively.

    For the converse direction, assume there exists an orientation of $H$ such that all terminal pairs in $\terminals$ are satisfied.
    Notice that due to the terminal pairs added in Step 1 and $h$ being bijective,
    for all $i \in [k]$ with $h(i) = (\alpha, \beta)$ there exist $j_1,j_2 \in [n]$ such that $\ell_\alpha \to x^i_{j_1}$ and $y^i_{j_2} \to r'_\beta$.
    Assume that $j_1 \neq j_2$, and notice that then the terminal pair $(x^i_{j_1}, y^i_{j_2})$ cannot be satisfied by our orientation, a contradiction.
    Consequently, let $s \colon [k] \to [n]$ be the function such that for all $i \in [k]$,
    $\ell_\alpha \to x^i_{s(i)}$ and $y^i_{s(i)} \to r'_\beta$, where $h(i) = (\alpha,\beta)$.
    We argue that $\mathcal{V} = \setdef{v^i_{s(i)}}{i \in [k]}$ is a $k$-clique of $G$.
    Assume that this is not the case, and let $v^{i_1}_{s(i_1)}, v^{i_2}_{s(i_2)} \in \mathcal{V}$
    such that $\{ v^{i_1}_{s(i_1)}, v^{i_2}_{s(i_2)} \} \notin E(G)$.
    In that case, it holds that $(x^{i_1}_{s(i_1)}, y^{i_2}_{s(i_2)}) \in \terminals$,
    however this terminal pair cannot be satisfied by the considered orientation, a contradiction.
    \end{claimproof}
\end{proof}

\subsection{XP algorithm}

In this section, we give an XP algorithm parameterized by $\vc$ with complexity
$n^{\bO(\vc^2)}$, matching our lower bound result.

\begin{theorem}\label{thm:algo_vc}
    There is an algorithm that solves {\SO} in time $n^{\bO(\vc^2)}$.
\end{theorem}

\begin{proof}

First we note that contracting all mixed cycles as in \Cref{prop:prepoc}
does not increase the size of a minimal vertex cover, so we can suppose that our graph is mixed acyclic.

Let $G$ be a graph of vertex cover number $\vc$.
Let $S\subseteq V$ be a vertex cover $G$ of size $\vc$. It can be computed in time $2^{\bO(\vc)}$ \cite{DBLP:conf/stacs/0001N24}.
First we guess the orientation of the edges of $G[S]$.
There are $2^{\bO(\vc^2)}$ possibilities.
Then, for all pairs $(u,v)\in S^2$, we guess if there is, in the final solution, an oriented path $u\to w\to v$ with $w\in V\setminus S$;
if we guessed there is at least one, we guess one such $w$, and orient its edges in consequence.
There are $n^{\bO(\vc^2)}$ possibilities.
As we will see in \Cref{lem:2path:vc}, after this guessing phase, as $S$ is a vertex cover,
the connectivity of $S$ in any solution respecting the guesses is already decided.
This will allow us to check the existence of such a solution in polynomial time.

\begin{claimrep}\label{lem:undirectededge:vc}
    Let $w\in V\setminus S$, $u,v\in S$ such that at the end of the guessing steps, there are two undirected edges $\{u,w\}$ and $\{v,w\}$. In any feasible solution respecting the guesses, either we orient them both from $w$ to $u,v$ or from $u,v$ to $w$.
\end{claimrep}

\begin{claimproof}
    First we observe that there are no mixed paths from $u$ to $v$ or from $v$ to $u$ that do not use $w$, as otherwise it would create a cycle, which contradicts the acyclicity of our solution. So, if we guessed that there is a path of length two from $u$ to $v$ or from $v$ to $u$ using $w'\in V\setminus S$, it must be that $w'=w$. 
    
    But also, as $\{u,w\}$ and $\{v,w\}$ are not oriented yet, it is not the case. So in any solution respecting the guesses, there is no path of length 2 from $u$ to $v$ or from $v$ to $u$. 
    
    Finally, the only possibility for orienting $\{u,w\}$ and $\{v,w\}$ without creating such paths is to orient both of them from $w$ to $u,v$ or from $u,v$ to $w$.
\end{claimproof}

The above claim shows that for any $w\in V\setminus S$, after the guessing steps, either all of its undirected edges must be oriented from $w$ to $S$ or from $S$ to $w$.

\begin{claimrep}\label{lem:2path:vc}
    Let $u,v\in V$; in any solution respecting the guesses, if there is a path from $u$ to $v$, then there is one that only uses already oriented edges after the guessing steps, except possibly for the first one if $u\in V\setminus S$, and the last one if $v\in V\setminus S$.
\end{claimrep}

\begin{claimproof}
    As $S$ is a vertex cover, for any $u,v\in V$, in a path $u=u_1\to u_2 \to\dots \to u_k=v$, there is no $u_i,u_{i+1}\in V\setminus S$. So, except for the the first edge if $u\in V\setminus S$ or the last edge if $v\in V\setminus S$, this path is a succession of $u_i\to u_{i+1}$ with $u_i,u_{i+1}\in S$, and those edges are already oriented after the guessing steps, or of $u_i\to u_{i+1}\to u_{i+2}$ with $u_i,u_{i+2}\in S$ and $u_{i+1}\in V\setminus S$, but we already guessed the existence of such paths of length 2, so in any feasible solution respecting the guesses, there is a path from $u_i$ to $u_{i+2}$ where the edges were already oriented after the guessing steps.
\end{claimproof}

This claim shows us that after the guessing steps, for any $(s_i,t_i)$ pair of
terminals that is not already satisfied, the only way to satisfy them while
respecting the guessing steps is to orient the edges of $s_i$ and $t_i$.

We now argue that there exists a demand pair $(s_i,t_i)$ which should be treated ``first''. To see this, consider an auxiliary directed graph which has a vertex for each demand pair, and an arc from $(s_i,t_i)$ to $(s_j,t_j)$ when $s_j=t_i$. If this auxiliary graph contains a directed cycle, this implies that the original instance must be oriented in a way that creates a directed cycle, but since we assumed that $G$ is mixed acyclic, this is impossible. Hence, the interesting case is when the auxiliary graph is a DAG, therefore there exists a demand pair $(s_i,t_i)$ such that for all 
other demands $(s_j,t_j)$ we have, $t_j\neq s_i$

We can apply the following method: take such a pair
$(s_i,t_i)$ (as described in the previous paragraph). If $s_i\in V\setminus S$, then orient all undirected edges between
$s_i$ and $S$ from $s_i$ to $S$. 
If the demand $(s_i,t_i)$ is now satisfied, that is, there is a path from $s_i$ to $t_i$ using arcs and oriented edges, we remove this demand.

If $(s_i,t_i)$ is still unsatisfied, if
$t_i\in V\setminus S$, then orient all undirected edges between $t_i$ and $S$
from $S$ to $t_i$. 

We have now handled the demand $(s_i,t_i)$: if it is satisfied after the steps above, we remove it from the instance and continue with the remaining demands; otherwise this demand cannot be satisfied and we reject.
We repeat this method until there
are no unsatisfied terminals.  

This method is correct: indeed, from
\Cref{lem:2path:vc}, after the guessing steps, the satisfaction of $(s_i,t_i)$
only depends on the orientation of the undirected edges incident on $\{s_i ,t_i\}\setminus S$. Furthermore, from the same claim, if $s_i\in V\setminus S$, in
any feasible solution, its undirected edges will only be used to satisfy pairs
of terminals that contain $s_i$. But our choice of terminals ensures that for
any such pair $(s_j,t_j)$, $s_i\neq t_j$, so $s_i=s_j$. So in a feasible
solution respecting the guessing steps, if an undirected edge of $s_i$ is used
to satisfy such a pair, it will be oriented from $s_i$ to $S$, so it was safe to select this orientation.  For the second step, if $(s_i,t_i)$ is still not
satisfied, then from \Cref{lem:2path:vc}, the only way to satisfy it while
respecting the guessing steps is, if $t_i\in V\setminus S$, to direct one
undirected edge between $t_i$ and $S$ from $S$ to $t_i$. But from
\Cref{lem:undirectededge:vc}, they all have the same orientation, so it was safe to orient
all of them from $S$ to $t_i$. If $(s_i,t_i)$ is still unsatisfied after these steps, then, from
\Cref{lem:2path:vc}, there is no way to satisfy those terminals while
respecting the guessing steps. 

The above method can be applied in time $n^{\bO(1)}$,
and must be applied at most $n^2$ times.  In the end, we have an algorithm of
complexity $n^{\bO(\vc^2)}$.  \end{proof}

\section{Parameterized by Number of Edges or Arcs}

Arguably one of the most natural ways to structurally parameterize \SO\ is to
consider either the number of undirected edges $|E|$ or the number of arcs
$|A|$ as a parameter. Our goal in this section is to investigate the
complexities of these two parameterizations.

We begin with the parameterization by $|E|$, which is trivially FPT: one can
simply guess the orientation of each edge and then check if this gives a valid
solution. The complexity is therefore $2^{|E|}n^{\bO(1)}$. Our contribution is
to observe that this trivial algorithm is perhaps optimal, as if one could
obtain an algorithm of complexity $(2-\varepsilon)^{|E|}n^{\bO(1)}$, then the
Strong Exponential Time Hypothesis (SETH) would be false. Recall that,
informally, the SETH states that there is no algorithm for \textsc{SAT} running
faster than the $2^n$ brute-force algorithm enumerating all assignments. We
show the lower bound by reusing a simple reduction from \textsc{SAT} which
appeared in \cite{ArkinH02}.

\begin{theoremrep}\label{thm:seth} If there exists an $\varepsilon>0$ such that
\SO\ can be solved in time $(2-\varepsilon)^{|E|}n^{\bO(1)}$, then the
satisfiability of a CNF formula $\phi$ with $n$ variables can be decided in
time $(2-\varepsilon)^n|\phi|^{\bO(1)}$, hence the SETH is false. \end{theoremrep}

\begin{proof}

Given a CNF formula $\phi$ with variables $x_1,\ldots,x_n$ we construct an
instance of \SO\ with $n$ undirected edges such that the new instance has a
solution if and only if the formula is satisfiable. Clearly, if we perform this
reduction (in polynomial time), then we obtain the theorem.

For each variable $x_i$ of $\phi$ we construct an edge in the new graph, with
one endpoint labeled $x_i$ and the other labeled $\neg x_i$. For each clause
$c$ of $\phi$, suppose that $c$ contains the literals
$\ell_1,\ell_2,\ldots,\ell_p$. We construct two vertices $s_c,t_c$ and the
terminal pair $(s_c,t_c)$. We add arcs from $s_c$ to the vertices labeled with
the \emph{negations} of all literals that appear in $c$; and we add arcs from
all vertices labeled with literals that appear in $c$ to $t_c$. This completes
the construction.

Correctness now can be seen as follows. If $\phi$ is satisfiable we orient the
edges so they point towards the true literal of each variable. Each clause $c$
contains a true literal, so there is an arc from $s_c$ to the vertex labeled
with the negation of this literal, then the edge is oriented to the vertex
labeled with the true literal, and then there is an arc to $t_c$, so the demand
$(s_c,t_c)$ is satisfied. For the converse direction, if there is a solution to
the \SO\ instance we extract an assignment to $\phi$ by setting a variable
$x_i$ to true if and only if the edge incident to the vertex labeled $x_i$ is
oriented toward this vertex. We claim this satisfies all clauses. Indeed, if a
clause $c$ were \emph{not} satisfied, this means that all edges are oriented
towards vertices labeled with the negations of the literals of $c$. This
implies that the demand $(s_c,t_c)$ is not satisfied, because any path starting
from $s_c$ must move towards a vertex labeled with the negation of a literal of
$c$ and since the undirected edges cannot be used to exit such a vertex we
cannot reach $t_c$. We conclude that we have a satisfying assignment of $\phi$.
\end{proof}

The main result of this section is then to establish that the parameterization
by $|A|$ also gives a fixed-parameter tractable problem, and indeed that we are
again able to obtain a single-exponential parameter dependence. This
parameterization is significantly more challenging, so we will need to rely on
structural properties obtained by any graph where $|A|$ is small after we apply
some simple reduction rules. The main result we obtain is stated below.

\begin{theorem}\label{thm:arcs} \SO\ admits a $2^{\bO(|A|)}n^{\bO(1)}$ algorithm.
\end{theorem}

Our high-level strategy will be made up of two parts. First, we define a (very)
restricted special case of \SO, where, notably, all undirected components are
paths (and we have some additional restrictions). With those severe
restrictions, this case becomes simpler. In particular, we will show that it admits an polynomial algorithm, via a reduction to \textsc{2-SAT}. Second, we
will show how a branching procedure can reduce any initial instance into a
collection of at most $2^{6|A|}=64^{|A|}$ instances of the restricted problem. The result will then follow by composing
the two algorithms. In order to ease presentation, we start with the algorithm
for the restricted case in \cref{sec:paths} and then give the algorithm for parameter $|A|$
in \cref{sec:arcs}.

\subsection{A Restricted Case}\label{sec:paths}

We define a restricted case of \SO\ as follows:

\begin{definition}\label{def:restricted} An instance $G=(V,E,A)$ of \SO\ is
called \emph{restricted} if we have the following: all vertices $v\in V$
of mixed degree at least $3$ are not incident on any unoriented edges, only arcs.
\end{definition}

Intuitively, a restricted instance which is also mixed acyclic is an instance
where each non-trivial undirected component is a path (because the maximum degree induced
by $E$ is $2$ and we have no cycles), the internal vertices of
each such path have no arcs connecting them to the outside, and furthermore each endpoint of such path has only one arc. In the following, we will show that these restricted instances are polynomial time solvable.

Our strategy to handle restricted instances of \SO\ will be the following:
first, we simplify our instance by making several observations allowing us to orient several edges deterministically. Then, we encode the obtained instance as a 2-SAT formula, by defining for each terminal in an undirected path a variable which encodes whether this vertex can reach (or be reached by) one endpoint of its path or the other in the final solution, and verifying that each terminal pair is satisfied.

\begin{lemmarep}\label{lem:restricted} There is an algorithm that takes a
restricted \SO\ instance $G=(V,E,A)$ that is mixed acyclic and decides if a
solution exists in time $n^{\bO(1)}$.  \end{lemmarep}

\begin{proof}

We begin with some simple observations. First, if we have a restricted \SO\
instance $G=(V,E,A)$ and we orient an undirected edge, that is, we produce a
new instance $G'=(V,E',A')$ by selecting an edge $e=\{u,v\}\in E$ and setting
$E'=E\setminus\{e\}$ and $A'=A\cup \{\ (u,v)\ \}$ or $A'=A\cup\{\ (v,u)\ \}$,
then the new instance is still restricted. This will be useful as we will first simplify the instance by observing that the orientation of some edges is forced.

Let $P_1,P_2,\ldots,P_c$ be the non-trivial components of $G[E]$ (that is, the ones with at least two vertices), which all induce
paths, as the instance is mixed acyclic and has maximum
degree $2$ in undirected edges. Suppose that there exists a terminal pair $(s,t)$
such that $s,t\in P_i$, for some $i\in[c]$. Then, there is a unique way to
satisfy this pair (as the graph is mixed acyclic), so we orient the edges of
$P_i$ connecting $s$ to $t$ appropriately and obtain a new instance. In the
remainder we assume that for each demand pair $(s,t)$ we have that $s\in P_i$,
$t\in P_j$ and $i<j$, which we can achieve by topologically sorting the
components.

Suppose now that there is a $P_i$ such that it has one incoming arc at its first endpoint and one outgoing arc at its second endpoint. We can observe that $P_i$ is oriented from its first endpoint to its second endpoint. Indeed, take $(s,t)$ a pair of terminals and a mixed path from $s$ to $t$ using $P_i$. If $s$ is in $P_i$, then $t$ is not from our previous simplification, so the mixed path uses the only outgoing arc from $P_i$, the one at its second endpoint, and this path does not conflict with the proposed orientation. The reasoning is symmetrical if $t\in P_i$. Otherwise, if $s,t\notin P_i$, then the only possibility is that it uses $P_i$ in its entirety from the first endpoint to the second one, which again, does not conflict with the proposed orientation. Therefore we orient it that way.

Finally, suppose we have $P_i$ such that $P_i$ has two incoming arcs and a source terminal $s$ in $P_i$. As its sink $t$ is not in $P_i$ and there is no mixed path beginning in $P_i$ and leaving it, the instance is trivially unsatisfiable. The reasoning is the same if $P_i$ has two outgoing arcs and one sink terminal.

We are reduced to the following instances: restricted \SO\ instances where each $P_i$ either has two outgoing arcs and only source terminals, or two incoming arcs and only sink terminals. Note that in that case if a terminal pair $(s,t)$ does not intersect a path $P_i$, there is no mixed path from $s$ to $t$ intersecting $P_i$. Therefore, if neither $s$ nor $t$ are in an undirected path, either there already is an oriented path from $s$ to $t$, and $(s,t)$ is already satisfied, or there is none and the instance is unsatisfiable. In the following, we suppose we do not have such pairs of terminals.

We will now reduce the problem to \textsc{2-SAT} in a way that if a solution
exists, the new instance is satisfiable; while if the new instance is satisfiable we can find some orientation satisfying all demands. This will be sufficient to solve the original instance.

For each non-trivial component $P_i$ of $G[E]$, we arbitrarily pick one endpoint of $P_i$ to be the left endpoint and call the other endpoint the right endpoint. For each terminal pair $(s,t)$, if $s$ is in a $P_i$, we introduce a variable $x_s$, and likewise, if $t$ is in a $P_i$, we introduce a variable $x_t$, whose informal meaning would be that if a variable is True, then its terminal will use the left endpoint of its path to join with its other terminal, and if it is false, it will use the right endpoint.
We now construct a \textsc{2-SAT} instance as follows:

\begin{enumerate}

\item (Path Consistency clause): For each $s_1, s_2$ (respectively $t_1,t_2$) in the path $P_i$ such that $s_1$ is strictly to the left of $s_2$ in $P_i$, we add the clause $x_{s_1}\lor \neg x_{s_2}$ (respectively $x_{t_1}\lor \neg x_{t_2} $)

\item (Path-Path Terminal Consistency clause): For each terminal pair $(s,t)$ such that $s$ is in $P_i$ and $t$ is in $P_j$ such that there is no oriented path between:
    \begin{itemize}
        \item the left endpoint of $P_i$ and the left endpoint of $P_j$, we add the clause $\neg x_s\lor \neg x_t$
        \item the right endpoint of $P_i$ and the left endpoint of $P_j$, we add the clause $x_s\lor \neg x_t$
        \item the left endpoint of $P_i$ and the right endpoint of $P_j$, we add the clause $\neg x_s\lor x_t$
        \item the right endpoint of $P_i$ and the right endpoint of $P_j$, we add the clause $x_s\lor x_t$
    \end{itemize}
\item (Path-Vertex Terminal Consistency Clause): For each terminal pair $(s,t)$ such that $s$ (respestively $t$) is in $P_i$ and $t$ (respectively $s$) is not in any non-trivial path, if there is no oriented path between:
    \begin{itemize}
        \item the left endpoint of $P_i$ and $t$ (respectively $s$ and the left endpoint of $P_i$), we add the clause $\neg x_s$ (respestively $\neg x_t$).
        \item the right endpoint of $P_i$ and $t$ (respesctively $s$ and the right endpoint of $P_i$), we add the clause $ x_s$ (respectively $x_s$).
    \end{itemize}

\end{enumerate}

This completes the construction and we now need to argue for correctness. 
For one direction, suppose that there exists an orientation satisfying all
demand pairs. For each terminal pair $(s,t)$, pick a directed path $P_{s,t}$ from $s$ to $t$ in the solution. If $s$ (respectively $t$) is in $P_i$ and uses the left endpoint of $P_i$, let $x_s=\textrm{True}$ (respestively $x_t=\textrm{True}$), and $\textrm{False}$ otherwise. 
This valuation satisfies the formula: indeed let $s_1,s_2$ be such that $s_1$ is strictly to the left of $s_2$ in $P_i$. If $P_{s_2,t_2}$ uses the left endpoint of $P_i$, then all edges between $s_2$ and the left endpoint are oriented towards the left endpoint, and therefore $P_{s_1,t_1}$ also uses the left endpoint. Therefore $x_{s_2}$ implies $x_{s_1}$ and the path consistency clause  $x_{s_1}\lor \neg x_{s_2}$ is satisfied (the reasoning is identical for sink terminals).
Furthermore, take $(s,t)$ such that $s$ is in $P_i$ and $t$ is in $P_j$. If $P_{s,t}$ uses endpoint $a$ of $P_i$ and endpoint $b$ of $P_j$, then there is a directed path in $G$ between $a$ and $b$, so, by definition, the Path-Path Terminal Consistency clauses concerning $x_s$ and $x_t$ are satisfied.
Finally, take $(s,t)$ such that $s$ (respectively $t$) is in $P_i$ and $t$ (respectively $s$) is not in any non-trivial path. If $P_{s,t}$ uses endpoint $a$ of $P_i$, then there is a directed path in $G$ between $a$ and $t_i$ (respectively $s_i$ and $x$), so, by definition, the Path-Vertex Terminal Consistency clauses concerning $x_s$ or $x_t$ are satisfied. 

For the converse direction, we claim that if the resulting instance has a
satisfying assignment, we can find a valid orientation of $E$ to satisfy all
demands. For each source terminal $s$ (respectively sink terminal $t$) in $P_i$, if $x_s=\textrm{True}$ (respectively $x_t=\textrm{True}$), orient all edges between $s$ (respectively $t$) and the left endpoint of $P_i$ from $s$ to the left endpoint (respectively from the left endpoint to $t$). If $x_s=\textrm{False}$, (respectively $x_t=\textrm{False}$), do the same with the right endpoint. After this, all remaining edges can be oriented in any direction. The Path-Path Terminal Consistency clauses ensure that this orientation is well-defined. Indeed, recall that each path $P_i$ contains either only source or only sink terminals. Consider the case of $P_i$ containing source terminals (the other case is identical) and observe that we could only have a contradiction if the above procedure tried to orient an edge $e$ to the right and the left at the same time. This would imply there is a source $s_1$ to the left of $e$ such that $x_{s_1}=\textrm{False}$ and a source to the right of $e$ such that $x_{s_2}=\textrm{True}$. This would, however, have falsified the clause constructed for $s_1,s_2$, so this cannot happen for a satisfying assignment. 

Take $(s,t)$ such that $s$ is in $P_i$ and $t$ is in $P_j$, $x_s=\textrm{True}$ and $x_t=\textrm{True}$. In the defined orientation, there is an oriented path from $s$ to the left endpoint of $P_i$, and from the left endpoint of $P_j$ to $t$, and the Path-Path Terminal Consistency clauses ensure that there is a directed path between both left endpoints (otherwise we would have added a clause falsified by this assignment). The reasoning is the symmetrical for other valuations of $x_s$ and $x_t$ and when only one of $s,t$ is in a non-trivial path $P_i$.

This completes the reduction and since \textsc{2-SAT} can be solved in
polynomial time we obtain the lemma.  \end{proof}

\subsection{FPT algorithm by Number of Arcs}\label{sec:arcs}

\begin{proof}[Proof of \cref{thm:arcs}]

We begin by applying \cref{prop:prepoc} and \cref{prop:remove:degree:1}
exhaustively. This does not increase $|A|$, so we can assume that the
input is mixed acyclic and has no vertices of degree at most $1$. Our algorithm
is now the following: take all vertices of mixed degree $3$ or more, and for each possible orientation of their edges, create a new instance. We claim that we obtain at most $2^{6|A|}$ restricted instances, such that at least one of them will have a solution if and only if the original instance did. 

First, we argue that if we perform this branching on all vertices of mixed degree $3$ or more, we obtain an instance conforming to \cref{def:restricted}. Indeed, after the orientation, all vertices of mixed degree $3$ are not incident on any unoriented edges, only arcs.

In order to bound the number of instances, let us count the number of edges to orient. We will need the following lemma:

\begin{claimrep}
    Let $T=(V,E)$ be an undirected forest, $V_3\subseteq V$ be the set of vertices of degree at least $3$ in $T$ and $V_1\subseteq V$ the set of vertices of degree $1$. Then, 
    $\sum_{v\in V_3}d(v)\le 3|V_1|$.
\end{claimrep}
\begin{claimproof}

To ease notation, when $T$ is a forest we will write $s_3(T)$ to denote the sum of the degrees of all vertices of degree $3$ or more in $T$ and $s_1(T)$ to denote the number of leaves of $T$. So, the claim is that $s_3(T)\le 3s_1(T)$ for all forests $T$.

It is sufficient to prove the claim for trees and we proceed by induction. If $|V|=1$ or $2$, the result holds. Otherwise, take a leaf $v\in V$, sharing an edge with $u$. By induction,  $s_3(T\setminus\{v\})\le 3s_1(T\setminus\{v\})$. 

We now take cases on the degree of $u$ in $T\setminus\{v\}$, which has to be at least $1$ (otherwise $T$ would have been a $K_2$).
If the degree of $u$ is: 

\begin{itemize}
        \item $1$ in $T\setminus\{v\}$, $s_3(T)=s_3(T\setminus\{v\})$ and $s_1(T)=s_1(T\setminus\{v\})$, so the result holds.
        \item $2$ in $T\setminus\{v\}$, then $s_3(T)=s_3(T\setminus\{v\})+3$ and $s_1(T)=s_1(T\setminus\{v\})+1$, so the result holds
        \item $3$ or more in $T\setminus\{v\}$, then $s_3(T)=s_3(T\setminus\{v\})+1$ and $s_1(T)=s_1(T\setminus\{v\})+1$
    \end{itemize}  
    In all cases the result holds.
\end{claimproof}

Armed with this claim, let us continue: take $T\subseteq G[E]$ composed of the non-trivial connected components of $G[E]$. As $G$ is mixed acyclic, it is a forest. Furthermore, as $G$ does not contain any leaf, each leaf of $T$ has at least one arc. Let $V_0$ be the set of leaves of $T$ with a unique arc, $V_1$ be the set of leaves of $T$ with two arcs or more,  $V_2$ the set of vertices with two edges and at least one arc, and $V_3$ the set of vertices with at least $3$ edges. We need to orient the unoriented edges of $V_1,V_2$ and $V_3$, which make $|V_1|+2|V_2|+\sum_{v\in V_3}d(v)$ edges. But, as each arc has only two endpoints, $|V_0|+2|V_1|+|V_2|\leq 2|A|$. And, from the previous claim, $\sum_{v\in V_3}d(v)\le 3(|V_0|+|V_1|)$. Therefore, $|V_1|+2|V_2|+\sum_{v\in V_3}d(v)\leq 3|V_0|+4|V_1|+2|V_2|\le 3(|V_0|+2|V_1|+|V_2|)\le 6|A|$:, so we need to orient at most $6|A|$ edges, thus obtaining at most $2^{6|A|}$ restricted instances.
\end{proof}

\section{Kernelization}
In this section we show that any instance $(G,\terminals)$ of \SO\ admits a polynomial kernel with respect to parameter $\vc+k$, where $\vc$ is the size of a minimum vertex cover of $G$ and $k$ is the number of terminal pairs.

In particular we will prove that:

\begin{theoremrep} \label{thm:poly:kernel:vc:and:terminals}
\SO\ admits a kernel of order $\bO(\vc^2 + k)$, where $\vc$ is the size of the minimum vertex cover of $G$ and $k$ is the number of terminal pairs. 
\end{theoremrep}

The high-level idea is that in the final kernel we will keep (i) the vertices of the vertex cover (which are $\bO(\vc)$) (ii) all terminals (which are $\bO(k)$) (iii) $\bO(\vc^2)$ additional vertices from the independent set. Similarly to the XP algorithm for parameter vertex cover, we use the idea that once we have fixed all paths of length at most $2$ between vertices of the vertex cover, this is sufficient to fully determine connectivity in the graph. Therefore, at most $\bO(\vc^2)$ vertices from the independent set are needed to form the solution. The challenge is then to identify a set of that size which is always sufficient and we achieve this via a reduction to \textsc{Maximum Matching}.

\begin{proof}
We start with some notation we will use throughout this section, followed by a description of the algorithm, followed by a proof of its correctness.

\subsection{Notation and Preprocessing}\label{sec:kernel:prelim}

We assume that the given graph is acyclic by exhaustively applying \cref{prop:prepoc}. Also we can assume that we are given a vertex cover $S$ of $G$. Since a 2-approximation to the minimum vertex cover can be computed in polynomial time, this assumption is without loss of generality (and can only affect the constants in the order of the produced kernel). In the remainder, we will use $\vc$ to denote the size of $S$ (slightly abusing notation).
We also set $I = V \setminus S$ be the corresponding independent set.  

Define $I_T \subseteq I$ to be the set of vertices that appear in a terminal pair, and let $I_N = I \setminus I_T$ be the remaining vertices of $I$.

If there exists a directed $u$–$v$ path in $(V,A)$ for some pair $u, v \in S$, then we add the arc $(u, v)$ to $A(G)$.  
Note that adding $(u,v)$ is safe, as it does not affect the feasibility of the original instance.
From now on, we assume that $G$ has been modified to include these arcs.

Let $X \subseteq S \times S$ be the set of (ordered) vertex pairs $(u, v)$ with $u \ne v$, such that:
\begin{itemize}
    \item $(u, v) \notin A(G)$, $\{u, v\} \notin E(G)$ and
    \item there exists a path $u - w - v$ in $G$ with $w \in I_N$.
\end{itemize}
Note that $|X| < 4\vc^2$. 

Assume we are given an orientation $\lambda$. We will say that pair $(u,v) \in S \times S$ \textit{appears} in a path $P$ of $G_{\lambda}$ if there exists a vertex $w \in I_N$ such that $u - w - v$ is a subpath of $P$.  
Moreover, if $(u,v)$ appears in $P$, we say that $(u,v)$ \textit{uses} the vertex $w$ in $P$ if $u - w - v$ is a subpath of $P$.

\subsection{Construction of the Kernel}\label{sec:construction:poly:kernel:vc:plus:k}

Suppose we have a vertex cover $S$, as explained above, and let $\vc$ denote the size of $S$. Also let $I$, $I_T$, $I_N$ and $X$ be the sets defined in the \cref{sec:kernel:prelim}.
The high-level idea behind the kernel construction is that the vertices in $I_N$ 
will only appear as intermediate vertices in paths connecting terminals. 
Specifically, they will be used to connect vertices of the vertex cover.  
Note that the only pairs that need such a vertex are those in $X$. 
Since the number of such pairs is bounded by $4\vc^2$ we will show that it suffices to keep a (specific) set of $I'\subseteq I$ of size $\binom{2\vc}{2}$ and delete the rest. The graph $G[S\cup I_T\cup I']$ will be the kernel.
We proceed with the selection of $I'$.

Let $|X| = \ell$, and $(u_j, v_j)$, for $j \in [\ell]$, be an enumeration of the pairs in $X$.

We now construct an undirected auxiliary graph $H = (V', E')$ as follows:
\begin{itemize}
    \item For each pair $(u_j, v_j) \in X$, create a vertex $x_j$.
    \item For each vertex $w \in I_N$, create a vertex $w'$.
    \item We define $V'$ as $\{x_j \mid (u_j, v_j) \in X\} \cup \{w' \mid w \in I_N\}$.
    \item We add an edge $\{x_j, w'\}$ to $E'$ if one of the following holds:
    \begin{itemize}
        \item $\{u_j, w\}, \{v_j, w\} \in E(G)$,
        \item $\{u_j, w\} \in E(G)$ and $(w, v_j) \in A(G)$,
        \item $(u_j, w) \in A(G)$ and $\{v_j, w\} \in E(G)$.
    \end{itemize}
\end{itemize}

We compute a maximum matching $M \subseteq E'$ in $H$.  
We say that a pair $(u_j, v_j)$ is \textit{matched} to a vertex $w \in I_N$ if $\{x_j, w'\} \in M$.  
Now, we define $I'$ using $M$ as follows; for each such edge $\{x_j, w'\} \in M$, we include the corresponding vertex $w$ in the set $I'$.  

Note that $|I'| <4 \vc^2$, as we include at most one vertex per pair in $X$ and $|X| < 4\vc^2$.  


Having selected $I'$, we have completed the construction of the kernel.  
In particular, we define the kernel to be the induced subgraph $G[S \cup I_T \cup I']$ along with the terminal pairs $\terminals$. We will denote this subgraph by $G^*$.  

\subsection{Equivalence of the two instances}

We will show that $(G^*, \terminals)$ is a yes-instance of \SO\ if and only if $(G, \terminals)$ is a yes-instance of \SO. We start with the forward direction.

\begin{claim} \label{lemma:kernel:forward:direction}
    If $(G^*,\terminals)$ is a yes-instance of \SO\ then $(G,\terminals)$ is a yes-instance of \SO.
\end{claim}

\begin{claimproof}
Observe that $G^*$ is an induced subgraph of $G$. 
Therefore, if there exists an orientation $\lambda$ of the edges of $G^*$ such that, for every terminal pair $(s, t) \in \terminals$, there exists a directed $s$–$t$ path in $G^*_\lambda$, then this orientation can be extended to all edges of $G$ by orienting the remaining edges arbitrarily.  
Hence, if $(G^*, \terminals)$ is a yes-instance of \SO, then so is $(G, \terminals)$.    
\end{claimproof}

We continue with the converse direction.
\begin{claim} \label{lemma:kernel:converse:direction}
    If $(G,\terminals)$ is a yes-instance of \SO\ then $(G^*,\terminals)$ is a yes-instance of \SO.
\end{claim}

Because this direction is rather technical, we give a sketch of the proof.

\paragraph*{Sketch of converse direction}

Let $\lambda$ an orientation of $E(G)$ such that, for every $(s,t) \in \terminals$, the graph $G_\lambda$ contains a directed $s$–$t$ path. 

For each $(s_i,t_i)\in \terminals$ let $P_i$ be a minimum $s_i$-$t_i$ path in $G_\lambda$. 
We show that, we can modify these paths so, any pair $(u,v) \in S\times S$ that appears in any of the paths uses the same vertex from $I$.

Let $X_P\subseteq X$ be pairs that appear in these paths. We proceed by proving that there exists a partition of $X_P$ and $I'\cup I_T$ into $X_1$, $X_2$ and $I'_1$ and $I'_2$, respectively, such that:
\begin{itemize}
    \item Each pair in $X_1$ uses a vertex from $I'_1$ in some path $P_i$, for $i \in [k]$;
    \item Each pair in $X_2$ is matched to a vertex in $I'_2$ by the matching $M$.
\end{itemize}

Assuming that $I'_1$ and $I'_2$ are given, we define an orientation $\mu$ of the edges in $E(G^*)$ as follows:
For each vertex $w \in I'_2$, we orient the edges incident to $w$ so that $(u_j, w) \in A(G^*_\mu)$ and $(w, v_j) \in A(G^*_\mu)$, according to the pair $(u_j, v_j)$ matched to $w$ by $M$.  
If the orientation of an edge incident to $w$ is not determined by the matching, we choose it arbitrarily.  
Finally, all remaining edges are oriented in $\mu$ as in $\lambda$.

Note that, given the above partitioning and the definition of $\mu$, the following holds:
\begin{itemize}
    \item For any pair $(u, v)$ appearing in a path $P_i$ with a subpath $v - w - u$:
    \begin{itemize}
        \item If $(u, v) \in X_1$, then $w \in I'_1$, and the subpath $v - w - u$ is directed in $G^*_\mu$ because $\mu$ and $\lambda$ agree on the orientation of edges incident to $I'_1$;
        \item If $(u, v) \in X_2$, then $(u, v)$ is matched to some $w' \in I'_2$, and this $w'$ can be used to replace $w$ in a directed $u$-$v$ path in $G^*_\mu$ by construction of $\mu$.
    \end{itemize}
\end{itemize}

This implies that the paths $P_i$ can be modified accordingly to form $s$-$t$ paths in $G^*_\mu$.
We now provide the detailed proof of the lemma.

\begin{claimproof}[Proof of \Cref{lemma:kernel:converse:direction}]
    we need to prove that given an orientation $\lambda$ of $E(G)$ such that, for every $(s,t) \in \terminals$, the graph $G_\lambda$ contains a directed $s$–$t$ path, we can compute an orientation $\mu$ of $E(G^*)$ such that, for every $(s,t) \in \terminals$, the graph $G^*_\mu$ contains a directed $s$–$t$ path.

Assume we are given such an orientation $\lambda$.  
For each $i \in [k]$, let $P_i$ be a shortest directed path from $s_i$ to $t_i$ in $G_\lambda$.

We say that a pair $(u,v) \in S \times S$ \textit{appears} in a path $P$ if there exists a vertex $w \in I_N$ such that the subpath $u - w - v$ is contained in $P$.  
Moreover, if $(u,v)$ appears in $P$, we say that $(u,v)$ \textit{uses} the vertex $w$ in $P$ if $u - w - v$ is a subpath of $P$.

Note that if a pair $(u,v)$ appears in any of the paths $P_i$, for some $i \in [k]$, then $(u,v) \in X$.  
Otherwise, either the arc $(u,v)$ would belong to $A(G^*)$ or the edge $\{u,v\}$ would belong to $E(G^*)$, contradicting the assumption that $P_i$ is a shortest path between $s_i$ and $t_i$.

The next step in proving the equivalence of the two instances is to establish the following claim.

\begin{claim} \label{claim:pairs:using:same:vertex}
    If a pair $(u,v) \in S \times S$ appears in any of the paths $P_i$, for $i \in [k]$, then $(u,v)$ can use the same vertex in all its appearances, in any path $P_i$, $i \in [k]$.
\end{claim}
    
\begin{claimproof}
    We proceed by fixing a pair $(u,v) \in S \times S$ that appears in some of the paths $P_i$, for $i \in [k]$.  
    We will modify the paths $P_i$, for $i \in [k]$, so that $(u,v)$ uses the same vertex in every path in which it appears.
    
    Consider a path $P \in \{P_i \mid i \in [k]\}$ such that $(u,v)$ appears in $P$, and let $w$ be the vertex that is used by $(u,v)$ in $P$.  
    We then identify every appearance of $(u,v)$ in all paths $P_i$, for $i \in [k]$, and replace the vertex it uses with $w$.
    
    After this modification, each $P_i$ remains a walk in $G_\lambda$.  
    We now argue that all modified $P_i$ are still paths. 
    Assume, that a $P_i$, for some $i \in [k]$, is no longer a path.  
    This implies that $w$ appears more than once in $P_i$.  
    Therefore, we can obtain a path $Q$ that is strictly shorter than $P_i$ and still connects $s_i$ to $t_i$.  
    This contradicts the assumption that the original $P_i$ was a shortest path.
    
    Finally, we repeat this process until every pair $(u,v)$ that appears in the paths $P_i$, for $i \in [k]$, uses the same vertex in all its appearances.
\end{claimproof}

Hereafter, we assume that the paths $P_i$, for $i \in [k]$, have been selected so that any pair that appears in any of the paths always uses the same vertex.  
Let $X_P$ denote the subset of pairs in $X$ that appear in at least one of the paths $P_i$, for $i \in [k]$.

The next step is to show how the vertices of $I'$ can be used to replace those vertices that are used by the pairs but do not belong to $I'$.

To do so, we first provide a partition of $I'\cup I_T$ and $X_P$ into two subsets each, $I'_1$, $I'_2$ and $X_1$, $X_2$, respectively, such that:
\begin{itemize}
    \item Each pair in $X_1$ uses a vertex from $I'_1$ in some path $P_i$, for $i \in [k]$;
    \item Each pair in $X_2$ is matched to a vertex in $I'_2$ by the matching $M$.
\end{itemize}

\proofsubparagraph{Partition of $I'\cup I_T$ and $X_P$:}

We begin by defining the set $X^1 \subseteq X_P$ such that a pair $(u_j, v_j)\in X^1$, $j\in [\ell]$, if and only if there is no edge $\{x_j, w'\} \in M$, where $M$ is the maximum matching of the auxiliary graph $H$. Observe that any pair $(u_j,v_j)\in X^1$ uses a vertex $w\in I'\cup I_T$. Indeed, if $(u_j,v_j)$ used a vertex $w\in I_N\setminus I'$, then $M\cup\{\{x_j,w'\}\}$ would be a matching in $H$ of size greater than $|M|$, which is a contradiction.

Let $I^1\subseteq I$ be the set of vertices used by the pairs $(u,v)\in X^1$. As we already proved, $I^1\subseteq I'\cup I_T$. Then let $X^2\subseteq X_P$ be the set of pairs matched to vertices in $I^1$ by the matching $M$ in $H$. Note that, since $I^1\cap I_T$ may be nonempty, not all vertices in $I^1$ are necessarily matched (by $M$) to pairs in $X_P$, but this does not affect the construction.
We will use $I^1$, $X^1$ and $X^2$ for an inductive construction. In particular, we will construct a sequence $(I^q,X^{q+1})$ such that, for any $q\ge 1$:
It holds that $I^q \subseteq I'\cup I_T$ and, for any pair $(u_j, v_j) \in X^{q+1} \setminus X^q$:
\begin{itemize}
    \item there exists a pair $(u_{j'}, v_{j'}) \in X^1$ such that $x_j$ and $x_{j'}$ are connected in $H$ via a path $Q$. Also, $Q$ alternates between edges in $M$ and $E(H) \setminus M$, and,
    \item the pair $(u_j, v_j)$ uses a vertex $w \in I' \cup I_T$ in every path $P_i$ it appears in.
\end{itemize}

First, we need to prove the properties for our base.

\textbf{Base Case} ($q=1$)\textbf{:}
Note that we have already proved that $I^1\subseteq I'\cup I_T$.
Consider any pair $(u_j,v_j) \in X^2 \setminus X^1$. By construction of $X^2$, $(u_j,v_j)$ has been matched by $M$ with a vertex $w\in I_1$. Since $w\in I_1$, there exists a pair $(u_{j'}, v_{j'}) \in X^1$ such that $(u_{j'}, v_{j'})$ uses $w$ in a path. Therefore, we can conclude that the edge ${x_{j'},w'}\in E(H)$. Since at least one of ${x_{j'},w'}$, ${x_{j},w'}$ is not included in $M$, we have a path $Q$ in $H$ that connects $x_j$ and $x_{j'}$.
Also, $Q$ alternates between edges in $M$ and $E(H) \setminus M$.

Now we will show that $(u_j,v_j)$ uses a vertex in $I'\cup I_T$ in all its appearances. Suppose, for contradiction, that the pair $(u_j, v_j)$ uses a vertex $v \in I_N \setminus I'$. Then we will use $Q$ to construct a matching in $H$ larger than $M$. Since $Q$ alternates between edges in $M$ and edges in $E(H) \setminus M$. Indeed, $M \setminus Q \cup (Q \setminus M) \cup { {x_j, v'} }$ forms a matching larger than $M$ in $H$. This contradicts the maximality of $M$, thus $v\in I'\cup I_T$.

\textbf{Inductive Hypothesis:}
Since we have proved that the invariants hold for $q=1$, we can assume that we have already constructed the sets $I^q$ and $X^{q+1}$ for some $q \geq 1$, and that $I^q$, $X^{q+1}$ satisfy the invariants. Now, we proceed with the inductive step.

\textbf{Inductive Step:}
Given the sets $I^q$ and $X^{q+1}$,
we construct the sets $I^{q+1}$ and $X^{q+2}$ as follows. 

Let $I^{q+1}$ be the set we obtain by adding to $I^q$ all vertices used by the pairs $(u,v)\in X^{q+1} \setminus X^q$. 
Notice that, by assumption $I^q\subseteq I'\cup I_T$ and the pairs in $X^{q+1}\setminus X^q$ use vertices in $I'\cup I_T$. Therefore $I^{q+1} \subseteq I' \cup I_T$. Then we define $X^{q+2}$ as the union of $X^{q+1}$ and all pairs $(u,v)\in X_P$ that are matched to vertices in $I^{q+1} \setminus I^q$.
It remains to show that $(I^{q+1},X^{q+2})$ satisfy the invariants. Note that $I^{q+1} \subseteq I' \cup I_T$ holds by the construction.

Now, consider any pair $(u_j, v_j) \in X^{q+2} \setminus X^{q+1}$. We will show that there exists a pair $(u_{j'}, v_{j'}) \in X^1$ such that there exists a path $Q$ that connects $x_j$ and $x_{j'}$ in $H$. Additionally, $Q$ alternates between edges of $M$ and $E(H)\setminus M$.

Since $(u_j, v_j) \in X^{q+2} \setminus X^{q+1}$, it is matched to some vertex $v \in I^{q+1}$, and there exists a pair $(u_{j''}, v_{j''}) \in X^{q+1} \setminus X^q$ such that:
\begin{itemize}
\item $(x_j, v') \in M$ (because $(u_j, v_j)$ is in $X^{q+2} \setminus X^{q+1}$), and
\item $(x_{j''}, v') \in E(H)$ (since the pair $(u_{j''}, v_{j''})$ must use $v$ and hence $v \in I^{q+1}$).
\end{itemize}
By the inductive assumption, there exists a path $Q$ in $H$ between $x_{j''}$ and $x_{j'}$, for some $(u_{j'}, v_{j'}) \in X^1$, alternating between edges in $M$ and $E(H) \setminus M$. Since $H$ is bipartite and $x_{j''}$ is an endpoint of $Q$, the edge incident to $x_{j''}$ in $Q$ belongs to $M$. Therefore, we can extend $Q$ to a path from $x_j$ to $x_{j'}$ by adding the edges $\{x_j, v'\} \in M$ and $\{x_{j''}, v'\} \in E(H)$, preserving the alternating property.

Next, we show that every pair $(u_j, v_j) \in X^{q+2} \setminus X^{q+1}$ uses a vertex $w \in I' \cup I_T$ in all paths $P_i$ it appears in. Suppose this is false and $w \in I_N \setminus I'$. Then, using the constructed alternating path $Q$ from the previous argument, the set $(M \setminus Q) \cup (Q \setminus M) \cup \{\{x_j, w'\}\}$ would be a matching in $H$ of size $|M| + 1$, again, a contradiction. This completes the induction step.

We terminate the construction when $I^{q+1} \setminus I^q = \emptyset$. Since $I^q \subseteq I^{q+1} \subseteq I'\cup I_T$, for every $q\ge 1$ this process will eventually terminate.

We then define $I'_1 = I^{q+1} \cup I_N$ and $I'_2 = I' \setminus I'_1$. Finally, we define a partition of $X_P$ as follows: let $X_{I_T}$ be the set of pairs $(u,v) \in X_P$ such that the vertex used by $(u,v)$ lies in $I_T$. Then, set $X_1 = X_{I_T} \cup X^{q+2}$ and $X_2 = X_P \setminus X_1$.

\proofsubparagraph{Properties of $X_2$ and construction of the new paths:}

Before proceeding, observe that any pair $(u,v) \in X_2$ has been matched by $M$ to a vertex $w \in I'_2$.  
Indeed, if $(u,v)$ had not been matched to a vertex in $I'_2$, then either it was not matched at all, or it was matched to a vertex in $I'_1$.  
In the first case, we would have $(u,v) \in X^1 \subseteq X_1$, contradicting $(u,v) \in X_2$.  
In the second case, the matched vertex belongs to $I^{q+1}$, implying $(u,v) \in X^{q+2} \subseteq X_1$, again a contradiction.

Next, we recall the definition of $\mu$ we gave in the sketch: For every vertex $w \in I'_1$, the orientation $\mu$ agrees with $\lambda$ on the edges incident to $w$.  
Then, for each vertex $w \in I'_2$ matched by $M$ to some pair $(u,v)$, we set $\mu$ to orient the edges incident to $w$ so that $(u,w)$ and $(w,v)$ are arcs in $A(G^*_\mu)$.  
This orientation is feasible due to the construction of the graph $H$ and the properties of the matching $M$.

We now claim that for all $i \in [k]$, the graph $G^*_\mu$ contains a directed $s_i$-$t_i$ path.  
To prove this, we modify the paths $P_i$ accordingly.  
Our goal is to ensure that every pair $(u,v) \in X$ appearing in the paths $P_i$ uses only vertices from $I'$, and moreover, that $\mu$ orients the relevant edges to allow such usage in $G^*_\mu$.

Fix a pair $(u,v)$ that appears in some $P_i$.  
If $(u,v) \notin X_2$, then it uses a vertex $w \in I'_1 \subseteq V(G^*_\mu)$.  
Since $\mu$ agrees with $\lambda$ on edges incident to these vertices, we conclude that the arcs $(u,w)$ and $(w,v)$ are present in $A(G^*_\mu)$.

Now consider the case where $(u,v) \in X_2$, and let $w$ be the vertex it uses in the paths.  
Since $(u,v)$ has been matched to a vertex $z \in I'_2$ by $M$, and $\mu$ orients the edges so that $(u,z)$ and $(z,v)$ are arcs in $G^*_\mu$, we can modify every path where $(u,v)$ uses $w$ by replacing $w$ with $z$.  
We apply this modification for all such pairs in $X_2$.

After these modifications, each $P_i$ becomes a walk in $G^*_\mu$ (since vertices may be repeated), but this suffices to guarantee the existence of $s_i$-$t_i$ paths for all $i \in [k]$.

This shows that if $(G, \terminals)$ is a yes-instance of \SO, then so is $(G^*, \terminals)$.
\end{claimproof}

Note that, the kernel presented in \cref{sec:construction:poly:kernel:vc:plus:k} together with \cref{{lemma:kernel:forward:direction}} and \cref{{lemma:kernel:converse:direction}} directly imply the correctness of \cref{thm:poly:kernel:vc:and:terminals} 

\end{proof}

\section{FPT by Treewidth plus Number of Terminals}

Our goal in this section is to show that even though \SO\ is NP-hard for graphs
of constant treewidth and also W[1]-complete parameterized by the number of
terminals $k$, it is in fact FPT when parameterized by the two parameters
together. In order to ease presentation we will avoid giving an argument based
on a dynamic programming algorithm and will prefer to instead rely on
Courcelle's theorem \cite{Courcelle90}. This will of course have the
disadvantage that we obtain an algorithm with running time of the form
$f(\tw,k)n^{\bO(1)}$, where the function $f$ is not explicitly bounded (and
could be a tower of exponentials). This is sufficient for our purpose of
classifying the complexity of this case as FPT, but we leave it as an
interesting open problem to determine the best parameter dependence $f$ one can
achieve.  We in fact establish the following slightly stronger claim:

\begin{theoremrep}\label{thm:courcelle} \SO\ is FPT parameterized by the treewidth
of the augmented underlying graph, that is, the treewidth of the graph obtained
by taking the underlying graph of the input and adding for each terminal pair
$(s_i,t_i)$ the edge $s_it_i$ to the graph. \end{theoremrep}

We observe that \cref{thm:courcelle} implies that \SO\ is FPT parameterized by
treewidth plus $k$ (the number of terminals), because adding the extra edges to
the underlying graph can increase its treewidth by at most $k$.

\begin{proof}[Proof of \cref{thm:courcelle}]

We observe that \SO\ can be expressed as an MSO$_2$ formula and then invoke
Courcelle's theorem \cite{Courcelle90}. Before we do this, it will be helpful
to recall the basics of what can be expressed in MSO$_2$ logic, which is the
logic of properties which can be expressed using quantification over sets of
vertices or edges of an input graph. It is important to note here that
Courcelle's theorem is able to handle directed (and indeed mixed)
graphs\footnote{Indeed, Courcelle's theorem holds also for directed
\emph{hypergraphs}, where hyperedges are ordered tuples of their endpoints
(cf.\ \cite{Courcelle90} p.\ 25), but in our case we only consider (normal)
edges of arity $2$.}.  More precisely, we can assume that we are given in the
input a graph $G=(V,E)$ and we are supplied with (i) a predicate
$\textrm{Edge}(e,x,y)$ which is satisfied if the edge $e$ is an edge from $x$
to $y$ (ii) unary predicates on the edges (or vertices) which allow us to
distinguish between a finite set of edge types.

We are now given an input to \SO\ consisting of a mixed graph $G=(V,E,A)$ and a
set of terminal pairs $\terminals = \setdef{(s_i,t_i) \in V \times V}{i \in
[k]}$. We can assume without loss of generality that $G[A]$ is acyclic (as
otherwise we can contract the directed cycle and simplify the instance). In
particular, $A$ does not contain any anti-parallel arcs.  We add to the graph
for each $i\in[k]$ an arc from $s_i$ to $t_i$. Furthermore, we replace each
$e=xy\in E$ with two arcs connecting $x,y$ in both directions.

We obtain a directed graph where edges have two types: type $1$ arcs are
original arcs ($A$) and arcs added to simulate edges of $E$; type $2$ arcs are
those added in the augmented graph.  Suppose that we have unary predicates
$T_1,T_2$ such that the meaning of $T_\alpha(e)$ is that the edge $e$ is of
type $\alpha\in\{1,2\}$. 

We can now construct an MSO$_2$ formula stating that the given \SO\ instance
has a solution. To do this we need to express that there exists a set of edges
$S$ of the input that satisfies the following:

\begin{enumerate}

\item All edges of $S$ are of type $1$.

\item For all $x,y\in V$, $S$ contains at most one arc with endpoints $x,y$
(that is, $S$ contains no anti-parallel arcs).

\item For all arcs of type $2$ with endpoints $s,t$, $S$ contains a directed
path from $s$ to $t$.

\end{enumerate}

The first two constraints are easy to express in MSO$_2$ (indeed, even in FO)
logic. The last (connectivity) constraint is well-known to be expressible in
MSO$_1$ logic as it is equivalent to the following property: for all sets of
vertices $X$, if $s\in X$ and $t\not\in X$, then there must exist an edge of
$S$ going from a vertex of $X$ to a vertex outside of $X$.   In particular, we
obtain the following formulas:

\begin{enumerate}

\item $\forall e\ \left(e\in S \to T_1(e) \right)$

\item $\neg \left( \exists e_1,e_2,x,y\ \left(e_1\in S\land e_2\in S\land
\textrm{Edge}(e_1,x,y) \land \textrm{Edge}(e_2,y,x) \right) \right)$

\item $\forall e,s,t\ \left(T_2(e)\land \textrm{Edge}(e,s,t)\right)\to $\\
$\neg \left(\exists X\ \left(s\in X\land t\not\in X\land \neg\left(\exists
e,a,b\ \left(e\in S\land a\in X\land b\not\in X\land
\textrm{Edge}(e,a,b)\right)\right)\right) \right)$

\end{enumerate}

It is now not hard to see that the formulas above indeed encode the desired
property, in particular because $A$ contained no anti-parallel arcs (hence the
second constraint only affects arcs added for undirected edges and ensures that
at most one direction is retained in $S$).  Observe that the MSO$_2$ formula we
have constructed has constant size.  \end{proof}

\section{Conclusion}\label{sec:conclusion}
Even though we have characterized the complexity of \SO\ for all possible values of treewidth (the only polynomial case is when the input is a forest), our NP-completeness only applies for graphs of pathwidth $3$: is the problem in P or NP-complete for graphs of pathwidth $2$. A similar question can be asked for feedback vertex set, where we show hardness when the feedback vertex set has size at least $2$: does the problem become tractable on graphs which are one vertex away from being forests? On the parameterized complexity side, even though we obtain a $2^{\bO(|A|)}n^{\bO(1)}$ algorithm, the base in the exponential is large ($64$), so improving this is a natural question.

\bibliography{bibliography}

\newpage

\appendix

\end{document}